\documentclass[12pt]{article}
\pdfoutput=1
\usepackage{algorithm}
\usepackage{algpseudocode}
\usepackage{array}
\usepackage[caption=false,font=normalsize,labelfont=sf,textfont=sf]{subfig}
\usepackage{textcomp}
\usepackage{stfloats}
\usepackage{url}
\usepackage[english]{babel}
\usepackage{graphicx}
\usepackage{framed}
\usepackage[normalem]{ulem}
\usepackage{amsmath}
\usepackage{amsthm}
\usepackage{amssymb}
\usepackage{mathtools}
\usepackage{amsfonts}
\usepackage{xfrac}
\usepackage{enumerate}
\usepackage{comment}
\usepackage[utf8]{inputenc}
\usepackage[top=1 in,bottom=1in, left=1 in, right=1 in]{geometry}

\usepackage{hyperref} 
\usepackage{cleveref}
\usepackage{verbatim} 
\usepackage{comment}

\usepackage{xcolor}

\usepackage{tikz}
\usepackage{pgf}
\usetikzlibrary{shapes.geometric, arrows, mindmap,shapes,positioning,backgrounds,decorations, matrix}
\usepackage{pgfplots}
\pgfplotsset{width=10cm,compat=1.9}
\usepgfplotslibrary{colormaps,fillbetween}

\theoremstyle{definition}
\newtheorem{Definition}{Definition}[section]

\theoremstyle{plain}
\newtheorem{Theorem}[Definition]{Theorem}

\newtheorem{Lemma}[Definition]{Lemma}

\theoremstyle{remark}
\newtheorem{Remark}[Definition]{Remark}

\numberwithin{equation}{section}

\usepackage[autostyle]{csquotes}
\usepackage[style=ieee,sorting=nyt]{biblatex}
\usepackage{authblk}

\DeclareMathOperator*{\dom}{dom}

\title{Inverse Problems are Solvable on Real Number Signal Processing Hardware}
\author[1,2,3]{Holger Boche}
\author[4]{Adalbert Fono}
\author[4,5,6]{Gitta Kutyniok}
\affil[1]{Institute of Theoretical Information Technology, TUM School of Computation, Information and Technology, Technical University of Munich, Germany}
\affil[2]{Munich Center for Quantum Science and Technology (MCQST), Munich, Germany}
\affil[3]{CASA – Cyber Security in the Age of Large-Scale Adversaries– Exzellenzcluster, Ruhr-Universität Bochum, Germany}
\affil[4]{Department of Mathematics, Ludwig-Maximilians-Universität München, Germany}
\affil[5]{Department of Physics and Technology, University of Tromsø, Norway}
\affil[6]{Munich Center for Machine Learning (MCML), Munich, Germany}
\date{}

\newcommand{\R}{\mathbb{R}}
\newcommand{\Q}{\mathbb{Q}}

\newcommand{\N}{\mathbb{N}}
\newcommand{\Z}{\mathbb{Z}}

\newcommand{\C}{\mathbb{C}}

\newcommand{\abs}[1]{\ensuremath{\left\lvert#1\right\rvert}}
\newcommand{\norm}[2][]{\ensuremath{\left\lVert#2\right\rVert_{#1}}}

\DeclareMathOperator*{\argmin}{arg\,min}

\begin{document}

\maketitle

\begin{abstract}
Despite the success of Deep Learning (DL) serious reliability issues such as non-robustness persist. An interesting aspect is, whether these problems arise due to insufficient tools or due to fundamental limitations of DL. We study this question from the computability perspective, i.e., we characterize the limits imposed by the applied hardware. For this, we focus on the class of inverse problems, which, in particular, encompasses any task to reconstruct data from measurements. On digital hardware, a conceptual barrier on the capabilities of DL for solving finite-dimensional inverse problems has in fact already been derived. This paper investigates the general computation framework of Blum-Shub-Smale (BSS) machines which allows the processing and storage of arbitrary real values. Although a corresponding real world computing device does not exist at the moment, research and development towards real number computing hardware, usually referred to by the term “neuromorphic computing”, has increased in recent years. In this work, we show that the framework of BSS machines does enable the algorithmic solvability of finite dimensional inverse problems. Our results emphasize the influence of the considered computing model in questions of accuracy and reliability. %Hence, the future of accurate and reliable DL may be closely connected to analog hardware.
\end{abstract}


\section{Introduction}
We study the field of inverse problems in imaging sciences from the perspective of algorithmic solvability respectively constructability of solutions. Image reconstruction from measurements is an important task in a wide range of scientific, industrial, and medical applications such as electron microscopy, seismic imaging, magnetic resonance imaging (MRI) and X-Ray computed tomography (CT). Various methods to solve inverse problems have been introduced, ranging from sparse regularization techniques \cite{Daubechies04SparseReg, Wright2009SparseRec, Cotter2005SparseSol, Selesnick2017SRviaCA} including compressed sensing \cite{Candes05DecLP, Candes06RobUnc, Candes06UnivEncStrat, Donoho06CompSens, Ji2008BayesianCS, Duarte2011StructuredCS, Elad2007ProjCS} to deep learning (DL) techniques \cite{Zu18AutoMap, Schlemper18CNNvsCS, Arridge2019SolvingIP, Bubba19Shearlet, Hammernik18MRIDL2, Borgerding2017AMP}.

In particular, DL is considered the predominant approach to tackle inverse problems nowadays. DL systems learn how to best perform the reconstruction by optimizing the quality of the reconstruction based on previous data. Although this process has been applied in various fields such as image classification \cite{He2015DelvingDI}, playing board games \cite{Silver16Go}, natural language processing \cite{Brown20GPT3}, and protein folding prediction \cite{Senior20DeepFold} with great success, there also do exist certain drawbacks of this approach, most prominently, the lack of reliability of DL systems \cite{Antun2020InstabilitiesDL}. In a larger context, lack of reliability refers to problems regarding safety, security, privacy and responsibility of DL systems \cite{Mireshghallah20PrivacySurvey, BOULEMTAFES20Privacy, Liu21Privacy, Willers20Safety, He2022Security}. Our main interest lies within the reliability and correctness of DL approaches. A typical counterexample is given by the instability towards adversarial examples, i.e., the fact that DL methods can be easily mislead through minor input perturbations constructed by an adversary \cite{Szegedy14AdvEx, Carlini18AudioAdvEx, Tsipras18RobustnessOdds}. Despite proposed methods to alleviate this instability phenomenon \cite{Papernot16Distillation, Madry18AdvTraining}, a full understanding is still missing \cite{Ilyas19AdvExNotBugs}. 

These disadvantages may be tolerable in static circumstances but in highly dynamic, safety-critical, and autonomous applications such as autonomous driving they pose a serious risk \cite{Liu2020ComputingSF, Muhammad2021AutonomVehicle2}, since there is in general no tolerance for error. Thus, enhancing reliability of DL is crucial to overcome the drawbacks. A way out is given by integrating a human observer into the DL pipeline \cite{WU22Hitl}. However, this may not be feasible in certain applications such as autonomous driving and also severely limits the autonomy of DL methods.

Another approach consists in certifying the correctness of DL systems for the desired applications \cite{Mirman21Certification, Biondi20Certification, Katz17Certification, Zhang20Certification, Salman19Certification}. Although this presents a huge challenge, it may point towards a solution. Typically, DL systems operate with no guarantees concerning the accuracy and correctness of their output. Hence, incorporating certificates with respect to accuracy and correctness is a step towards reliable DL. However, implementing this process in a DL pipeline is far from trivial, even more it is not clear if it is possible at all. A first step is to assess whether and under which circumstances a solution to a given problem can be reliably obtained by an algorithmic computation. If the answer is positive, then in the second step the possibility of implementing certificates can be evaluated. On the other hand, if the answer is negative, then also the certificate-based approach and thereby a reliable DL approach in this strict sense can not be expected. In this way, we are able to distinguish between inherent limitations of DL itself and fundamental limitations of any algorithm tackling a given problem. 

Specifically, we study whether the solution of inverse problems --- a task often tackled by DL ---  can in principle be computed. For this purpose, we apply the notion of algorithmic solvability, which describes an abstract framework for accurate computations on a given computing device. We find that the algorithmic solvability of inverse problems depends on the utilized computing device; on digital hardware algorithmic solvability can not be achieved, whereas on analog hardware it is potentially feasible. Thus, obtaining reliable DL tackling inverse problems may ultimately depend on the computing platform. 

\subsection{Algorithmic Solvability}
\noindent
Given a mathematical formulation of a specific problem, an immanent question is whether a solution exists and, upon existence, how to obtain it. Naturally, the second task depends on the feasible algorithmic operations which are prescribed by the properties of the utilized computing device and its hardware. Moreover, the constructive solution may not be exact but only approximate. For instance, on digital computers real numbers can in general only be represented approximately by rational numbers (since displaying an irrational number would require an infinite amount of information). Consequently, we can not expect exact solutions of real problems on digital hardware. At the same time, an approximate solution can be accepted if guarantees regarding accuracy, convergence, and worst-case error are satisfied. 

It is important to realize that existence and construction of solutions are separate questions under certain circumstances. Although the existence of a solution can be proved by explicitly constructing it, equally valid is an approach that deduces the existence logically without exactly specifying the solution. Consider, for example, the well-known Bolzano-Weierstrass theorem which states that every bounded sequence in $\R^n$ has a convergent subsequence. Based on this result, we can assert the existence of a convergent subsequence for any bounded sequence in $\R^n$; however, at the same time we may not be able to write down the convergent sequence explicitly.

In practice, typically a constructive solution is required. In general, we aim to implement the approach on a computer so that the calculations can be performed autonomously. More formally, we have the following dependencies and connections: Given a task or problem expressed in some formal language, which defines the premises of the underlying system, an algorithm is a set of instructions that solves the posed problem under these premises. The algorithm can only operate within clearly defined boundaries which exactly describe its admissible steps. Often, an algorithm is intended to run on a hardware platform which also can only perform clearly prescribed operations. Hence, for an algorithmic solution on a given hardware platform the admissible operations of the algorithm and the hardware need to coincide. In other words, the capabilities of the utilized computing device dictate the admissible operations of the (sought) algorithm. Therefore, for practical applications the intended hardware clearly defines the capabilities and limitations of any suitable algorithm. 

Hence, it is crucial to be aware of the limitations of a given hardware platform concerning computing capabilities. In this manner, we can evaluate whether reliable DL is in principle possible. If an algorithmic solution of a specific problem is feasible on a given hardware, then reliable DL is potentially achievable on this hardware platform. 

Another related topic is the complexity of an algorithm. Here, not only the mere existence of an algorithm is studied but the efficiency (measured in the number of computation steps) of algorithms. Then, one can classify algorithms according to their complexity and decide which suffice practical demands with respect to computation time, memory requirements, etc. However, the question of complexity and practical applicability is posed only after the algorithmic solvability is established.

\subsection{Analog and Digital Computing Platforms}\label{seq:AnaDigComp}
\noindent
The most prominent and universally applied computing and signal processing platform is digital hardware. Turing machines represent an abstract concept of digital machines and thereby allow to deduce limitations of digital computing devices \cite{Turing36Entscheidung}. Consequently, they provide a means to study the question whether or not it is in principle possible to compute a solution of a task algorithmically. %solve a task algorithmically (without putting any constraints on the computational complexity of such an algorithm).
Problems, which are not computable by a Turing machine, cannot be solved or even algorithmically approximated in a controlled way by any current or future digital computer architecture. One of the key feature of Turing machines as well as of digital computers is their discrete nature, i.e., Turing machines operate on a discrete domain and enable the execution of discrete algorithms. In contrast, many real-world problems are of a continuum nature, i.e., their input and output quantities are described by continuous variables. Hence, these quantities can only be approximated but not represented exactly on Turing machines. This dichotomy lies at the core of the non-computability of many tasks.

More precisely, Turing machines are by construction limited to work with computable numbers. The concept of computability was first introduced by Borel in 1912 \cite{Borel1912Comp}. On this basis, Turing refined the concept by linking it to an abstract computing device --- the Turing machine --- and established notions like computable numbers and computable functions on computable real numbers \cite{Turing36Entscheidung}. Computable numbers are real numbers that are computable by a Turing machine in the following sense. The result of computations by a Turing machine is by construction always a rational number (since Turing machines need to stop their calculations after a finite amount of steps). Therefore, computable real numbers --- which constitute a countable, proper subset of the real numbers --- are those that can be effectively approximated by computable sequences of rational numbers, i.e., by Turing machines.

In the framework of Turing machines the algorithmic solvability of finite-dimensional inverse problems was studied in \cite{Boche2022LimitsDL}. It was found that any method which runs on digital hardware is subject to certain boundaries when approximating the solution maps of inverse problems. Hence, this result establishes a fundamental barrier on digital computing devices that also effects DL implemented on digital hardware. This barrier was obtained by using the mathematical structure and properties of finite-dimensional inverse problems as well as the mathematical structure and properties of Turing machines and thereby also of digital machines.

An interesting question is whether this constraint is connected to the properties of Turing machines, in particular, the approximation of computable real numbers by rational numbers, or, does this constraint arise from inherent characteristics of inverse problems? In other words, is the result problem-specific or connected to digital hardware? Therefore, it is of interest to study how powerful the signal processing unit must be to enable the algorithmic solution of inverse problems. In the following, we consider and analyse inverse problems under a more general (analog) computation model that is based on exact real number calculations, i.e., arbitrary real numbers can be processed and stored. Thus, such a model is not suitable for implementation on digital hardware and it is even unclear if such a computing device can be realized by means of today's (hardware) technology. 

One might argue that treating computing models beyond the standard digital hardware is purely of theoretical interest since such hardware platforms can not be obtained with current (or even possibly future) manufacturing capabilities. However, in recent years novel approaches such as neuromorphic computing and signal processing systems have been proposed \cite{Christensen2022NCSurvey}. Regarding a practical development and implementation of electronic neuromorphic hardware platforms, there has been made intriguing progress in industrial research, among others by IBM \cite{WebsiteIBM}, Intel \cite{WebsiteIntel} and Samsung \cite{Ham2021SamsungNMC}. Neuromorphic systems are inspired by biological neural networks. They can be regarded as a combination of digital and analog computation, which is achieved by incorporating real values in form of electrical values, like current and voltages, instead of solely relying on binary numbers. 

Apart from incorporating analog processing capabilities %the enhanced computing power (in comparison with digital hardware)
neuromorphic hardware offers further possible advantages. The expected energy savings from deploying artificial intelligence (AI) applications (as deep learning) on neuromorphic computing and signal processing hardware, instead of applying digital implementations, are tremendous \cite{Esser15Neuromorphic, Smith22Neuromorphic, Maas22Neuromorphic, Markovic20Neuromorphic, Blouw20Neuromorphic}. Consequently, neuromorphic computing and signal processing platforms are considered as the forthcoming premium solution for implementing AI applications \cite{Papp2021NanoscaleNN}. Furthermore, neuromorphic computing and signal processing offers by design advantages for emerging concepts like "in-memory computing" compared to digital computing and signal processing platforms \cite{Boybat21InMemory,Karunaratne20InMemory, Sebastian20InMemory,Payvand19InMemory}. In particular, these concepts can lead to further substantial energy savings in neuromorphic approaches. Due to these reasons, extensive scientific and engineering research is carried out in industry to develop neuromorphic hardware platforms which already led to the aforementioned impressive result.

Another line of research focuses on biocomputing where living cells act as computing and signal processing platforms which allow analog computations for human-defined operations \cite{Grozinger2019Biocomp}. In recent years, tremendous progress has been made in understanding and implementing analog computing features in biological systems \cite{Wagenbauer2017DNAassembly, Poirazi2017Dendritic, Wright2022DeepPhysycalNN}, however, still many challenges lie ahead. Nevertheless, these examples show that analog information processing hardware, e.g., based on neuromorphic computing processors, may ultimately be within reach.

Lastly, we also want to mention the progress in quantum simulations which provide yet another approach to analog computations; see \cite{Bloch22Quantum, Bloch22Quantum2} and the references therein. In contrast to digital computers via Turing machines, there does not exist a general mathematical concept to describe analog computations universally. The Blum-Shub-Smale (BSS) machine \cite{Blum89BSSmachines} is a suitable model to analyse the limits of analog computations; potentially, BSS machines yield a mathematical description of a universal analog computer based on exact real number computation. It was conjectured in \cite{Grozinger2019Biocomp} that the mathematical model of BSS machines provides an (abstract) description of biocomputing and neuromorphic systems. However, it is not clear whether and to what degree a realization is feasible in practice. Thus, the BSS model is not apt to evaluate the capabilities of current analog hardware. On the other hand, the BSS model enables us to investigate whether idealized analog, i.e., exact real number processing, hardware allows for the algorithmic solution of a given problem. If a problem is not algorithmically solvable on a BSS machine, then the problem is deemed inherently hard and it is very unlikely that it can be solved on any (future) analog hardware device. 

Therefore, an interesting question is whether the same limitations concerning inverse problems as in the Turing model also arise in the BSS model. Intriguingly, we show that the situation is quite different under the BSS model. In fact, here the opposite is true, namely, the solution maps of inverse problems can under certain conditions be computed without restrictions. Hence, on advanced neuromorphic computing devices the fundamental boundaries, which arise on digital hardware, may not exist. This shows that the algorithmic solvability of inverse problems is inherently connected to the computing platform.

\subsection{Related Work}
\noindent
First, we want to point out that the capabilities of BSS machines compared to Turing machines have already been studied for tasks such as denial-of-service attacks \cite{Boche2021DoSAttacks, Boche2021DetectDoS} and remote state estimation \cite{Boche2022RemoteState}. 

For the special case of inverse problems, we already mentioned the non-computability result in the Turing model in \cite{Boche2022LimitsDL}. A similar result was obtained in \cite{colbrook21stable,bastounis21extended} for Oracle Turing machines, where Turing machine gain access to arbitrary real numbers via rational Cauchy sequences provided by an oracle. This is a clear distinction to Turing machines which operate on a proper subset of the real numbers --- on the computable numbers. An implication is that non-computability results in the Oracle Turing model do not necessarily translate to Turing machines --- only the reverse is true. However, the results in \cite{bastounis21extended, colbrook21stable} also transfer to the Turing model. Under the Oracle framework also the computability of various problems in learning, regularisation and computer-assisted proofs was analyzed \cite{bastounis21extended}. In our inverse problem application a non-computability statement holds in both models. For a detailed discussion about the differences we refer to \cite{Boche2022LimitsDL}. 

The idea behind the oracle approach is that the input representation in a computing device will be necessarily inexact. Either the inputs themselves are corrupted by noise, measurement error, rounding, etc. or the hardware is not able to represent the input exactly as is the case with digital hardware, i.e., Turing machines, and real numbers, in general. Moreover, floating-point arithmetic, a common approach to represent numbers on today's digital hardware, stores even rational number such as $1/3$ only approximately. To model this setting, coined inexact input, the computing device receives inputs (via the oracle) in form of rational Cauchy sequences approximating the "true" values. Thus, the input representation is detached from the actual computing device and any notion of computability on specific hardware assumes the described input representation. In this sense, the inexact input model aims for a general description of a given task independent (to a certain degree) from the utilized computing device.

In \cite{colbrook21stable, bastounis21extended} the inexact input model was not only applied to Turing machines (which yielded the aforementioned non-computability statement for Oracle Turing machines) but also to BSS machines. In other words, the capabilities of BSS machines for solving inverse problem was studied under inexact input representations (via an oracle). It was shown that under this assumption the same non-computability result holds for BSS machines as for Turing machine. It is important to stress that the found computational boundaries are inherently connected to the inexact input model but not necessarily to the power of the computing model.

In contrast, our goal is to assess the intrinsic capabilities and limitations of BSS machines to solve inverse problems. %in their intended use case.
From our perspective, the core strength of BSS machines --- to store and process real numbers precisely without error --- is neglected when using the inexact input model. In particular, the input representations are sequences of rational numbers and therefore BSS machines behave very similar to Turing machines. Indeed, the differences between Turing and BSS machines come mainly into effect on irrational numbers which are excluded by the approximative input representations based on rational numbers. Thus, it is not surprising that in the inexact input model the same restrictions towards computability of inverse problems hold for BSS machines as for Turing machines. But at the same time not the full power of BSS machines is taken into consideration. Our findings clearly show that if BSS machines are allowed to process real numbers exactly --- instead of working with rational approximations --- then they have strictly greater capabilities than Turing machines for solving inverse problems.     

\subsection{Results and Impact}
\noindent
We consider the algorithmic solvability of finite-dimensional inverse problems on BSS machines and compare the findings to known results in the Turing model. For real-valued inverse problems we establish a clear disparity between capabilities of BSS and Turing machines. The implications for analog and digital computing devices are noteworthy. In general, analog hardware modeled by BSS machines allows for a solver which can be applied to any inverse problem of a given dimension. Thus, the solver is not connected to a specific inverse problem but is able to handle inverse problems of fixed dimension simultaneously. In contrast, on digital hardware, the currently predominant hardware device, the algorithmic solvability is limited such that the described solver can not exist. Nevertheless, less powerful solvers adapted to a specific inverse problem may still be realizable on digital hardware. However, they lack the universality of the general solver on analog hardware since they can be only applied to a fixed, pre-defined task. Hence, approaches to solve real inverse problems on upcoming analog hardware in form of neuromorphic devices may potentially have strictly greater capacity than current digital solutions.

In the more general case of complex-valued inverse problem the situation is more intricate. Here, the Turing and BSS model share limitations concerning the algorithmic solvability of finite-dimensional inverse problems. However, modifications in the problem formulation suffice to guarantee BSS solvability whereas the negative statement in the Turing model remains true. In particular, the algorithmic solvability in the BSS model of an approximate problem is proved. Consequently, for the modified problem formulations we can establish the same differences between analog and digital devices concerning algorithmic solvability as in the real case. This shows in principle that the Turing and BSS model behave differently and have distinct computational barriers in the field of complex-valued inverse problem.      

The implications of our findings concerning DL are two-fold. DL implemented on digital hardware is bounded by the limitations of the hardware itself for solving inverse problems. Therefore, reliable DL may only be realized in restricted settings, where only a limited set of inverse problems is admissible. In contrast, idealized analog hardware in principle allows for reliable DL systems tackling any inverse problem of fixed dimension simultaneously. Thus, DL systems implemented on neuromorphic hardware may potentially enable optimal and trustworthy reconstructions in inverse problems in very general circumstances.

\subsection{Outline}
\noindent
A concise overview of real number computing theory will be presented in \Cref{sec:RNC}. This is followed by an introduction to DL  with a particular focus on the inverse problems in \Cref{sec:DLforIP}. Next, in \Cref{sec:DLonBSS} the capabilities of DL on real number processing hardware modeled by BSS machines is studied. Thereby, we also establish the underlying problem of the algorithmic solvability of finite-dimensional inverse problems. \Cref{sec:CompRecMap} covers the formal statement of our main result about solvability of inverse problems in the real number computing model. Finally, a discussion about the implications of our finding on today's solvers --- in particular deep learning methods --- of inverse problems finishes the paper in \Cref{sec:Discussion}.

\section{Real Number Computing}\label{sec:RNC}
\noindent
In this section we will formally introduce a computing model over the real numbers as well as its properties.

\subsection{Blum-Shub-Smale machines}
\noindent
In 1989, Blum, Shub, and Smale proposed in \cite{Blum89BSSmachines} a general computing model over an arbitrary ring or field $R$: the Blum-Shub-Smale (BSS) machines. It is the basis of algebraic complexity theory. Not only is the BSS model defined for finite fields, but also allows to carry over important concepts from classical complexity theory in the Turing machine model to computational models over a larger variety of structures, e.g., infinite fields such as $\R$. Here, a BSS machine can store arbitrary real numbers, can compute all field operations on "$\R$", i.e., "$+$" and "$\cdot$", and can compare real numbers according to the relations "$<$", "$>$", and "$=$". BSS machines provide therefore the mathematical basis for real number signal processing. Thus, in principle, the BSS model is suitable to investigate the power of analog information processing
hardware like neuromorphic computing processors.%a neuromorphic machine has the same power as a BSS machine.

Turing machines that are able to work with arbitrary real numbers are so-called Oracle Turing machines \cite{Ko91ComplTheoryRealFunc}. An oracle provides for each real number $x$ a Cauchy sequence $(a_n)_{n\in\N}$ of rational numbers that controls the approximation error, i.e., $\abs{x-a_n}<2^{-n}$ for $n\in\N$. Then, the input to the Turing machine is this sequence of rational numbers representing the corresponding real number $x$ and the Turing machine performs its operation on the given sequence. %Then, the Turing machine obtains from the oracle this sequence of rational numbers as input and works with this sequence of rational numbers as a representation for the corresponding real number $x$. 
We wish to remark that BSS machines differ from such Oracle Turing machines.

In essence, BSS machines can be considered a generalization of Turing machines. If $R$ is chosen to be $\Z_2 = \{ \{0,1\},+,\cdot \}$, then BSS machines recover the theory of Turing machines. Moreover, a BSS machine is similar to a Turing machine in the sense that it operates on an infinite strip of tape according to a so-called program. This is a finite directed graph with five types of nodes associated with different operations, namely input node, computation node, branch node, shift node, and output node. For every admissible input, the output of a BSS machine is calculated according to the program in a finite number of steps, i.e., the BSS machine finishes its program in finite time and stops. For a detailed introduction and description of BSS machines and programs running on BSS machines, we refer the reader to \cite{Blum04CompoverReals, Blum98ComplRealComp} and references therein.

\begin{Definition} 
    \textit{BSS-computable functions} are input-output maps $\Phi$ of the BSS machine $\mathcal{B}$. The output $\Phi_{\mathcal{B}}(x)$ is defined if the BSS machine $\mathcal{B}$ terminates on input $x$ and the output is generated by the program of the BSS machine $\mathcal{B}$. %, i.e., for every input $x$, the output $\Phi_{\mathcal{B}}(x)$ is defined if the output is reachable by the program of the BSS machine $\mathcal{B}$.
\end{Definition}

\begin{Definition}
    A set $A \subset \R^N$ is \textit{BSS-decidable} if there exists a BSS machine $\mathcal{B}_A$ such that, for all $x \in \R^N$, we have $\mathcal{B}_A(x) = \chi_A(x)$, i.e., the characteristic function $\chi_A$ of the set $A$ is BSS-computable.\\
    A set $A \subset \R^n$ is \textit{BSS-semidecidable} if there exists a BSS machine $\mathcal{B}_A$ such that, for all $x \in A$, we have $\mathcal{B}_A(x) = 1$.
\end{Definition}

In this paper, we consider BSS machines and study whether or not they are capable of computing the reconstruction maps of inverse problems. We prove that these reconstruction maps are indeed computable by BSS machines showing that real number signal processing enables algorithmic solvability of this class of problems. 

To establish our result we utilize the connections between BSS machines and semialgebraic sets, i.e., sets that can be described by finitely many polynomial equations and inequalities.

\subsection{Semialgebraic Sets}
\noindent
The concept of semialgebraic sets, defined by polynomial equations and inequalities, is key for BSS-decidability. Every BSS-semidecidable set in $\R^n$ can be expressed as a countable union of semialgebraic sets \cite{Blum98ComplRealComp}. This follows from the algebraic structure of the available basic operations in BSS-algorithms. For an introduction and discussion of semialgebraic sets, we refer to \cite{Bochnak98RealAlgGeom}. We follow mainly \cite{Basu98AlgRealAlgGeom} when introducing this topic.

\begin{Definition}
    The \textit{class of semialgebraic sets} in $\R^n$ is the smallest class of subsets of $\R^n$ that contains all sets $\{x \in \R^n : p(x) > 0\}$ with real polynomials $p : \R^n \to \R$ and is in addition closed under finite intersections and unions as well as complements.    
\end{Definition}
\begin{Remark}
    Any semialgebraic set in $\R^n$ can be expressed as the finite union of sets of the form $\{x \in \R^n \mid p(x) = 0 \wedge \bigwedge_{q \in \mathcal{Q}} q(x) > 0\}$ with a finite set $p \cup \mathcal{Q}$ of real polynomials $p,q : \R^n \to \R$.  
\end{Remark}

The Tarski-Seidenberg theorem \cite{Tarski51ProjThm, Seidenberg54ProjThm} states that any projection map on $\R^n \to \R^m$, where $n \geq m$, projects semialgebraic sets in $\R^n$ onto semialgebraic sets in $\R^m$. A consequence is that quantifier elimination is possible over $\{\R,+,\cdot,0,1,>\}$. This means that every first-order formula built over $\{\R,+,\cdot,0,1,>\}$, i.e., a formula involving polynomials with logical operations "$\wedge$", "$\vee$", and "$\neg$" and quantifiers "$\forall$", "$\exists$", can be algorithmically transformed into an equivalent quantifier-free formula. Thus, there exists an algorithm that, given an arbitrary formula (with quantifiers) as input, computes an equivalent formula without quantifiers. This allows us to eliminate all quantifiers in semialgebraic sets algorithmically. 

For instance, the projection theorem implies that a semialgebraic set defined by formulas of the form 
\begin{equation*}
    \{(x_1,\dots,x_n) \in \R^n : \exists x_{n+1} \text{ such that } p(x_1,\dots,x_n,x_{n+1})\, \Delta \,0\}
\end{equation*}  
can be rewritten as a semialgebraic set defined by formulas of the form 
\begin{equation*}
    \{(x_1,\dots,x_n) \in \R^n : q(x_1,\dots,x_n)\, \Delta\, 0\},    
\end{equation*}
where $\Delta \in \{<,>,\leq,\geq,=\}$ and $p,q$ are real polynomials. 

Furthermore, Tarski found an algorithm to decide the truth of sentences, i.e., formulas that have no free variables, in the first order language built from $\{\R,+,\cdot,0,1,>\}$ (cf. \cite{Cucker2019CompSemiAlg}). This algorithm directly results from the transformation of semialgebraic sets with quantifiers into semialgebraic sets without quantifiers.

The complexity of the elimination process has been reduced significantly since the original algorithm proposed by Tarski, see \cite{Cucker2019CompSemiAlg} and references therein. However, for many applications in information theory the computational complexity of quantifier elimination still remains too high. In case of computability results the complexity of the algorithm is not relevant (just its existence), hence quantifier elimination provides a very useful theoretical tool. Tarski’s algorithm can be computed on a BSS machine. In contrast, the elimination of quantifiers is in general not possible on Turing machines.  %provides a reason that this decision problem is not solvable on Turing machines.

Many problems in real geometry can be formulated as first-order formulas. Tarski-Seidenberg Elimination Theory can then be used to solve these problems. For our needs, the crucial observation is that, based on quantifier-elimination, it is possible to develop an algorithm for finding %the infimum of a polynomial on a semi-algebraic set as well as 
a minimizer of a polynomial on a semialgebraic set if there exists one. Here, an algorithm refers to a computational procedure that takes an input and after performing a finite number of admissible operations produces an output. The feasible operations are exactly the ring or field operations of the considered structure and, additionally, comparisons between elements if the structure is ordered. Thus, by construction such an algorithm can be executed on a BSS machine.

For further details such as the procedure of the algorithm, we refer to \cite{Basu98AlgRealAlgGeom}. We only state the input and output relation of the algorithm for real polynomials for the convenience of the reader:

\begin{algorithm}[H]
\caption{Global Optimization}\label{alg:GlOp}
\begin{algorithmic}
    \Require  a finite set $\mathcal{P}$ of real polynomials $p:\R^n \to \R$ describing a (non-empty) semialgebraic set $S$ by a quantifier free formula $\Phi$ and a polynomial $f:\R^n\to\R$.
    \Ensure the infimum $w$ of $f$ on $S$, and a minimizer, i.e. a point $x^\ast \in S$ such that $f(x^\ast) = w$ if such a point exists.
\end{algorithmic}
\end{algorithm}

The non-emptiness of a semialgebraic set can be decided by an algorithm in the BSS model. This is again a special case of quantifier elimination. Hence, we can check whether the optimization domain is empty before calling the global optimization algorithm. Let us mention that these capabilities of BSS machines also highlight the differences to Turing machines. The analogous discrete problem, commonly known as Hilbert's tenth problem \cite{Hilbert00Problems}, of deciding whether a polynomial equation with integer coefficients and a finite number of unknowns possesses integer-valued roots is not solvable on Turing machines \cite{Matiyasevich70Diophantine}.

\section{Deep Learning for Inverse Problems}\label{sec:DLforIP}
\noindent
In this section, we give a short introduction to DL with a particular focus on solving inverse problems \cite{Jin17DCNNInvProb, Adler17DNNInvProb, Ongie20DLInvProb, Schlemper18CNNvsCS, Mousavi15DLvCSSigRec}. For a comprehensive depiction of DL theory we refer to \cite{Goodfellow16DL} and \cite{Berner2021modernMathDL}.

\subsection{Inverse Problems}
\noindent
We consider the following finite-dimensional, underdetermined linear inverse problem:
\begin{equation}\label{eq:problem}
    \text{Given noisy measurements }  y = Ax + e \in \C^m \text{ of } x \in \C^N, \text{ recover } x,
\end{equation}
where $A \in \C^{m \times N}, m< N$, is the \textit{sampling operator} (or measurement matrix), $e \in \C^m$ is a noise vector, $y \in \C^m$ is the \textit{vector of measurements}, and $x \in \C^N$ is the object to recover (typically a vectorized discrete image). Classical examples from medical imaging are magnetic resonance imaging (MRI), where $A$ encodes the Fourier transform, and computed tomography (CT), where $A$ encodes the Radon transform. In addition, in practice the underdetermined setting $m < N$ is common, the reason being that in many applications the number of measurements is severely limited due to time, cost, power, or other constraints.

There exist different approaches to solve inverse problem, a particularly successful is given by deep learning. The area of inverse problems in imaging sciences was to a certain extent swept by deep learning methods in recent years. A unified framework for image reconstruction by manifold approximation as a data-driven supervised learning task is proposed in \cite{Zu18AutoMap}. In \cite{Arridge2019SolvingIP}, the authors survey methods solving ill-posed inverse problems combining data-driven models, and in particular those based on deep learning, with domain-specific knowledge contained in physical–analytical models. Approaches to tackle specific inverse problems have been presented for instance in \cite{Bubba19Shearlet} for limited angle CT, in \cite{Yang16MRIDL, Hammernik18MRIDL2} for MRI, in \cite{Chen2018LowLightPhoto} for low-light photography, in \cite{Rivenson17DLmicroscopy} for computational microscopy, and in \cite{Araya18DLtomography} for geophysical imaging.

\subsection{Basics of Deep Learning}
\noindent
The inspiration for DL comes from biology, as it utilizes an architecture called (artificial) \textit{neural network} mimicking the human brain. A neural network consists of a collection of connected units or nodes subdivided into several layers allowing an artificial neural network to learn several abstraction levels of the input signal. In its simplest form an \textit{L-layer feedforward neural network} is a mapping $\Phi : \R^d \to \R^m$ of the form
\begin{equation}\label{eq:NNdef}
    \Phi(x) = T_L \rho(T_{L-1}\rho(\dots \rho(T_1 x))), \quad x\in \R^d ,
\end{equation}
where $T_{\ell} : \R^{n_{\ell- 1}} \to \R^{n_{\ell}}$, $\ell=1,\dots,L$, are affine-linear maps
\begin{equation*}
    T_{\ell} x = W_{\ell} x + b_{\ell}, \quad W_{\ell} \in \R^{n_{\ell} \times n_{\ell-1}}, b_{\ell} \in \R^{n_{\ell}},
\end{equation*}
$\rho: \R \to \R$ is a non-linear function acting component-wise on a vector, and $n_0 = d, n_L = m$. The matrices $W_{\ell}$ are called \textit{weights}, the vectors $b_{\ell}$ \textit{biases}, and the function $\rho$ \textit{activation function}. A neural network can easily be adapted to work with complex inputs by representing the inputs as real vectors consisting of the real and imaginary parts. 

Thus, a neural network implements a non-linear mapping parameterized by its weights and biases. The primary goal is to approximate an unknown function based on a given set (of samples) of input-output value pairs. This is typically accomplished by adjusting the network’s parameters, i.e., its weights and biases, according to an optimization process; the standard approach so far is stochastic gradient descent (via backpropagation \cite{Rumelhart86BP}). This process is usually referred to as the training of a neural network.

Triggered by the drastic improvements in computing power of digital computers and the availability of vast amounts of training data, the area of DL has seen overwhelming practical progress in the last fifteen years. Deep neural networks (which in fact inspired the name "deep learning"), i.e., networks with large numbers of layers, lead to several breakthroughs in many applications \cite{He2015DelvingDI, Silver16Go,Brown20GPT3,Senior20DeepFold}. Moreover, the current trends towards neuromorphic hardware and more broadly towards biocomputing (see \Cref{seq:AnaDigComp}) promise further fundamental developments in this area. Therefore, the necessity for a thorough analysis of the mathematical foundations of DL is immanent.

\section{Deep Learning on BSS machines}\label{sec:DLonBSS}
\noindent
An obvious question is whether a neural network can be implemented on a BSS machine, in particular, can the network be evaluated on a given input. If this is not the case, then certainly BSS machines do not provide the necessary tools to improve DL in comparison with implementations on Turing machines. Therefore, the input-output relation of the network has to be computable, i.e., the network has to be a BSS-computable function. Next, we show that this is indeed the case under some mild assumptions.

\begin{Theorem}
    A neural network $\Phi$ as defined in \eqref{eq:NNdef} is a BSS-computable function given that the activation function $\rho: \R \to \R$ is BSS-computable.
\end{Theorem}
\begin{proof}
    The network $\Phi$ is a concatenation of affine linear maps represented by matrix-vector multiplications and vector additions and non-linearities represented by the activation function. Affine linear maps are BSS-computable functions, since they are formed exclusively through basic field operations (addition and multiplication) that can be executed on the computation nodes of a BSS machine. Since the activation function $\rho$ is assumed to be BSS-computable, there does exist a program of a BSS-machine that generates the input-output relation of $\Phi$ which is equivalent to saying that $\Phi$ is BSS-computable.
\end{proof}
\begin{Remark}
    The most common and universally applied activation function is the ReLU activation $\text{ReLU}(x) = \max\{x,0\}$, which is merely a branch condition depending on a comparison, i.e., it can be performed on the branching nodes of a BSS machine. Hence, ReLU activation is a BSS-computable function.   
\end{Remark}

Since the evaluation of a given neural network is feasible on a BSS machine the subsequent question is whether a specific neural network can be obtained for solving inverse problems, i.e., can a neural network be trained on a BSS machine to solve inverse problems? In the next section, we will formalize this question precisely.

\subsection{Problem Setting}\label{sec:InvProb}
\noindent
Deep learning techniques can be incorporated in the solution approach of inverse problems in various ways. Depending on the properties of a given inverse problem or the aspired application, a specific utilization of DL in the reconstruction process may seem the most promising. We focus on an end-to-end approach, where the goal is to directly learn a mapping from measurements $y$ to reconstructed data $x$; see  \cite{Ongie20DLInvProb} for alternative approaches which employ deep learning at certain steps in the processing pipeline, e.g., in order to learn a regularizer. The end-to-end approach represents the most fundamental method, since it requires no further problem specific knowledge or assumptions.

Then, the goal of the training process is to obtain a deep neural network $\Phi$, which approximates the mapping from measurements to the original data. An important question is whether $\Phi$ exists, and upon existence, can the network be constructed by an algorithm? In other words, does there exist a training algorithm that computes suitable weights and biases such that $\Phi$ performs the intended reconstruction reliably? Such a training algorithm can only exist if the underlying inverse problem is algorithmically solvable on the applied hardware platform. Note that mere algorithmic solvability does not imply the existence of a practically applicable and efficient algorithm to solve inverse problems. Therefore, applying DL to algorithmically solvable problems is still beneficial since they ideally provide a generic and straightforward approach in practice. On the other hand, algorithmic non-solvability clearly indicates limitations that even DL can not circumvent so that e.g. a trade-off between generality and reliability can not be avoided.   

A general solution strategy for inverse problems is to rewrite the model \eqref{eq:problem} in a mathematically more tractable form since \eqref{eq:problem} is in general ill-posed. Aiming to account for uncertainties of the measurements, a relaxed formulation is considered, which has a considerably simpler solution map than the original description \eqref{eq:problem}. Typically, the goal is to express \eqref{eq:problem} as an optimization problem given a sampling operator $A \in \C^{m\times N}$ and a vector of measurements $y \in \C^m$. There exist various formulations of this optimization problem, a straightforward one is given by the least-squares problem
\begin{equation}\label{eq:lsq}
    \argmin_{x\in C^N} \norm[\ell_2]{Ax-y}. \tag{ls}
\end{equation}
A minimizer of the least-squares problem can be straightforwardly obtained via the pseudoinverse of $A$. Since there exists an algorithm in the BSS model that computes the pseudoinverse of an arbitrary matrix (see \Cref{sec:appb}, in particular proof of \Cref{thm:BPapprox}), we can conclude that a minimizer of \eqref{eq:lsq} can be algorithmically computed on BSS machines. Conversely, such an algorithm can not exist on Turing machines \cite{Boche2022PseudoInverse}. However, the solution of \eqref{eq:lsq} is generally not unique and the individual solutions tend to display different qualitative properties, i.e., not all minimizers of \eqref{eq:lsq} are of the same value when reconsidering the original problem \eqref{eq:problem}. Therefore, additional regularization terms are added to the optimization problem which impose desired characteristics on the solution set. This results in minimization problems such as (quadratically constrained) \textit{basis pursuit} \cite{Candes06Stable, Chen1998AtomicDec} 
\begin{equation}\label{eq:sparseprob}
    \argmin_{x \in \C^N} \norm[\ell_1]{x} \text{ such that } \norm[\ell_2]{Ax -y} \leq \varepsilon \tag{bp}
\end{equation}
and unconstrained \textit{lasso} \cite{Belloni11SquareRootLasso, Tropp06Relax, Lv2011GroupLasso}
\begin{equation}\label{eq:lasso}
    \argmin_{x \in \C^N} \lambda \norm[\ell_1]{x} + \norm[\ell_2]{Ax -y}^2, \tag{la}
\end{equation}
where the magnitude of $\varepsilon>0$ and $\lambda>0$ control the relaxation. The underlying idea is to exploit sparsity in the recovery without explicitly forcing sparse solutions via the $\ell_0$ norm, which is typically intractable in many applications.  

In order to reconstruct data from any given measurement we need to solve an optimization task described in \eqref{eq:sparseprob} or \eqref{eq:lasso}. An algorithmic solution can by definition only exist if the reconstruction map --- from measurements $y$ to corresponding data $x$ --- is computable on a given computing device. The actual algorithm to evaluate the reconstruction map may still be unknown, but by proving the computability of the map we can claim its existence, or, in case the map is not computable, exclude its existence. Note that there still might exist algorithms that are able to solve a problem to a certain degree in practice although the problem is not computable. Depending on the hardware platform, the algorithm might only work properly for a restricted class of inputs, for a certain accuracy, or without correctness guarantees. Hence, understanding the limitations of a given computing device is a necessity if one aims to evaluate the solvability of a certain problem.

Recall, that computability of a function $f$ in the BSS model is equivalent to the existence of a BSS machine $\mathcal{B}$ such that its input-output map $\Phi_{\mathcal{B}}$ satisfies $f(z) = \Phi_{\mathcal{B}}(z)$ for any $z$ in the domain of $f$. We introduce the following multi-valued functions to study the computability of inverse problems. For fixed sampling operator $A  \in \C^{m \times N}$ and some fixed optimization parameter $\mu > 0$ %, e.g., $\varepsilon$ and $\lambda$ in \eqref{eq:sparseprob} and \eqref{eq:lasso}, respectively, 
denote by 
\begin{align}\label{eq:Xi_Aeps}
    \Xi_{O,A,\mu}: \C^{m} &\rightrightarrows \C^N \\
    y &\mapsto P(A,y,\mu) \nonumber %\argmin_{x \in \C^N} \norm[\ell_1]{x} \text{ such that } \norm[\ell_2]{Ax -y} \leq \varepsilon \nonumber
\end{align}
the (multi-valued) reconstruction map so that $\Xi_{P,A,\mu}(y)$ represents the set of minimizers for the optimization problem $P(A,y,\mu)$ given a measurement $y \in \C^m$. %%\eqref{eq:sparseprob} given a measurement $y \in \C^m$.  
For instance, for basis pursuit $O(A,y,\mu)$ is described in \eqref{eq:sparseprob} with $\mu =\varepsilon$.

The function $\Xi_{P,A,\mu}$ allows us to examine the algorithmic solvability of a specific inverse problem determined by the considered optimization problem $P$ with fixed optimization parameter $\mu$ and the sampling operator $A$. Thus, a hypothetical algorithm should take a measurement $y \in \C^m$ as input and yield the corresponding set of reconstructions $\Xi_{P,A,\mu}(y)$. However, we are not only interested in approximating the reconstruction map of a specific inverse problem associated to the fixed sampling operator $A$. Instead, the key is that we aim to approximate the reconstruction map of any inverse problem of fixed dimension; the reason being that we are interested in obtaining an algorithm which may be applied to an arbitrary inverse problem described by the optimization $P$, without adjusting it to specific properties of an individual case. The reconstruction map of these inverse problems with fixed dimension $m\times N$ and parameter $\mu>0$ is given by
\begin{align}\label{eq:Xi_mNeps}
    \Xi_{P,m,N,\mu}: \C^{m\times N} \times \C^m &\rightrightarrows \C^N \\  
    (A,y) &\mapsto P(A,y,\mu). \nonumber %\argmin_{x \in \C^N} \norm[\ell_1]{x} \text{ such that } \norm[\ell_2]{Ax -y} \leq \varepsilon \nonumber
\end{align}
Thus, $\Xi_{P,m,N,\mu}(A,y)$ represents the set of minimizers for the optimization problem $P(A,y,\mu)$ given a sampling operator $A  \in \C^{m \times N}$ and an associated measurement $y \in \C^m$, i.e.,
\begin{equation*}
    \Xi_{P,m,N,\mu}(A,y) = \Xi_{P,A,\mu}(y).    
\end{equation*}
To clarify the difference to \eqref{eq:Xi_Aeps}, note that a potential algorithm takes the sampling operator $A \in \C^{m \times N}$ as additional input. Therefore, the algorithm is required to succeed for any sampling operator of dimension $n \times N$ associated to an inverse problem described by $P$. 

We also consider the even more general case where the optimization parameter $\mu$ is not fixed but also part of the input. Hence, let 
\begin{align}\label{eq:Xi_mN}
    \Xi_{P,m,N}: \C^{m\times N} \times \C^m \times \R_{>0} &\rightrightarrows \C^N \\
    (A,y,\mu) &\mapsto P(A,y,\mu) \nonumber %\argmin_{x \in \C^N} \norm[\ell_1]{x} \text{ such that } \norm[\ell_2]{Ax -y} \leq \varepsilon \nonumber
\end{align}
with $\Xi_{P,m,N}(A,y,\mu)$ representing the set of minimizers of the optimization problem $P(A,y,\mu)$ given an optimization parameter $\mu>0$, a sampling operator $A  \in \C^{m \times N}$, and an associated measurement $y \in \C^m$. 

Observe that the mappings \eqref{eq:Xi_Aeps}, \eqref{eq:Xi_mNeps}, and \eqref{eq:Xi_mN} are in general set-valued, since the solution of the considered optimization problems, e.g., \eqref{eq:sparseprob} and \eqref{eq:lasso}, does not need to be unique. However, computability of multi-valued maps is a stronger result than required for algorithmic solvability in practical applications. Indeed, computing all feasible solutions is much harder than computing just one. In practice, in most circumstances the user is not interested in the general case of the whole solution set, but it suffices to obtain exactly one feasible solution. Thus, a hypothetical algorithm should take a sampling operator $A \in \C^{m \times N}$, a measurement $y \in \C^m$, and a optimization parameter $\mu>0$ as input and yield exactly one corresponding reconstruction $x \in \Xi_{P,m,N}(A,y,\mu)$. Therefore, also single-valued restrictions of the original multi-valued mappings are significant, because they represent this simpler problem.

To formalize this concept, note that for a multi-valued function $f:\mathcal{X} \rightrightarrows \mathcal{Y}$ there exists for each input $x \in \text{dom}(f)$ at least one output $y_x^\ast \in f(x) \subset \mathcal{P}(\mathcal{Y})$. A single-valued restriction of $f$ can then be defined as the function
\begin{align*}
    f^s: \mathcal{X} &\to \mathcal{Y} \\
    x &\mapsto y_x^\ast.    
\end{align*}
We denote by $\mathcal{M}_f$ the set of all the single-valued functions associated with the multi-valued function $f$. Hence, the set $\mathcal{M}_f$ encompasses all single-valued functions $f^s$ that are formed by restricting the output of a multi-valued map $f$ to a single value for each input. If there exists at least one function in $\mathcal{M}_f$ that is computable, then we can algorithmically solve the problem proposed by $f$.

\begin{Definition}\label{def:AlgSolv}
    A problem with an input-output relation described by a multi-valued function $f:\mathcal{X} \rightrightarrows \mathcal{Y}$ is \textit{algorithmically solvable on a BSS machine or Turing machine} if there exists a function $f^s \in \mathcal{M}_f$ that is computable on a BSS or Turing machine, respectively. 
\end{Definition}
In particular, it is not relevant which of the (possibly infinitely many) functions in $\mathcal{M}_f$ is computable, since any of those is an appropriate solution. Even more, it may be the case that most functions in $\mathcal{M}_f$ are non-computable; but as long as we succeed to show computability for just one of them, we consider the task algorithmically solvable. For our needs, the relevant sets of functions are $\mathcal{M}_{\Xi_{P,m,N,\mu}}$ and $\mathcal{M}_{\Xi_{P,m,N}}$ corresponding to the multi-valued functions $\Xi_{P,m,N,\mu}$ and $\Xi_{P,m,N}$ introduced above, respectively. \Cref{img:AlgSketch} now illustrates a comparison and summary of the described settings. 

\begin{figure}[h]
    \centering
\scalebox{0.8}{
\begin{tikzpicture}[->,>=stealth',thick] 

    \matrix (m) [matrix of nodes,row sep=0.7cm,column sep=0.25cm, ampersand replacement=\&,  nodes={align=center, text width =3.25cm}]{
    \& idealised practical application \& \& minimal theoretical requirement \& {or \\ equivalently}\\
    \& {inverse problem \\ $(A,y,\mu)$} \& \& {input \\ $(A,y,\mu)$} \& {$(A,y,\mu)$}\\
    \& |[draw]|{algorithm \\ $\mathcal{A}$} \& \& |[draw]| {BSS machine \\ $\mathcal{B}$} \& \\
    {measurement \\ $y$} \& |[diamond, draw, inner sep=-1ex,aspect=2.5]| {solver} \& {reconstr. of $x \in$ \\ $\Xi_{P,m,N}(A,y,\mu)$} \& {output $x \in$ $\Xi_{P,m,N}(A,y,\mu)$} \& {$\Xi^s(A,y,\mu)$, \\ $\Xi^s \in \mathcal{M}_{\Xi_{P,m,N}}$} \\
    };
    
    \path[shorten >=0.05cm] (m-2-2) edge (m-3-2);
    \path[shorten >=0.05cm, shorten <=0.05cm] (m-3-2) edge (m-4-2);
    \path[shorten >=0.05cm]  ([xshift=-10pt,yshift=9pt]m-4-1.east) edge (m-4-2.west);
    \path[shorten <=0.05cm] (m-4-2.east) edge ([xshift=5pt,yshift=9pt]m-4-3.west);
    
    \path[shorten >=0.05cm] (m-2-4) edge (m-3-4);
    \path[shorten >=0.05cm, shorten <=0.05cm] (m-3-4) edge (m-4-4);
    
    \path[shorten >=0.05cm, shorten <=0.05cm] (m-2-5) edge node [right, align=center, text width =2cm] {BSS-comp. \\ function} (m-4-5);
    
    \draw[decorate,decoration={snake,post length=2mm}] (m-1-2) to  (m-1-4);
    \draw[<->,decorate,decoration={snake,post length=1.5mm, pre length=2mm}] ([xshift=-10pt]m-1-4.east) to  ([xshift=21pt]m-1-5.west);
\end{tikzpicture}
}
\caption{Our goal is to study the existence of an algorithm $\mathcal{A}$ on an analog computing device which tackles inverse problem described by an optimization process $P$. In particular, $\mathcal{A}$ takes the sampling operator $A\in \C^{m\times N}$, measurement $y \in \C^m$, and optimization parameter $\mu>0$ and generates a solver which reconstructs data $x \in \Xi_{P,m,N}(A,y,\mu)$ given measurement $y$. In an abstract mathematical description, the setting boils down to the existence of a BSS machine $\mathcal{B}$ taking inputs $(A,y,\mu)$ and outputting $x$. The existence of this BSS machine requires that there exists at least one BSS-computable function $\Xi^s \in \mathcal{M}_{\Xi_{P,m,N}}$.
}
\label{img:AlgSketch}
\end{figure}
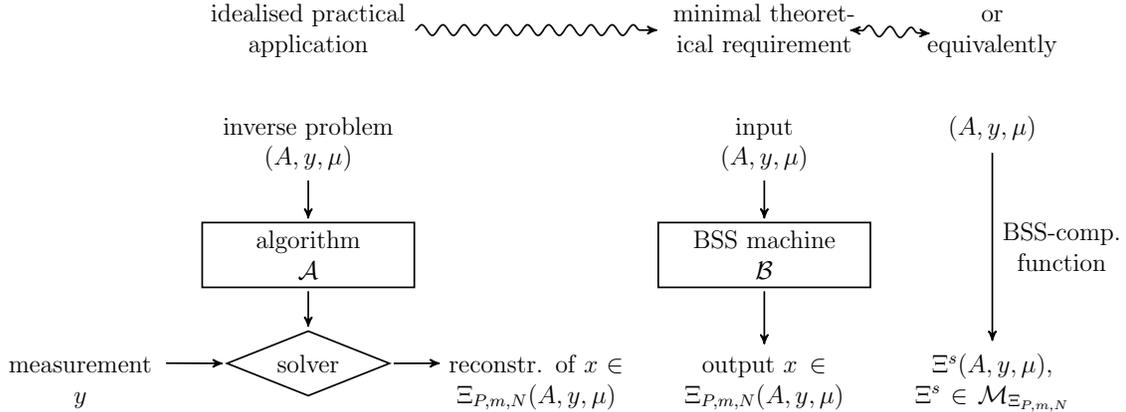

\section{Computability of the Reconstruction Map}\label{sec:CompRecMap}
\noindent
Our goal is to study BSS-computability of the reconstruction maps introduced in \eqref{eq:Xi_mNeps} and \eqref{eq:Xi_mN} for basis pursuit \eqref{eq:sparseprob} and lasso optimization \eqref{eq:lasso}. 
First, we consider basis pursuit optimization. In the Turing model the non-computability of any element of $\mathcal{M}_{\Xi_{\text{bp},m,N,\varepsilon}}$ has already been established \cite{Boche2022LimitsDL} --- which also immediately implies the non-computability of the multi-valued map $\Xi_{\text{bp},m,N,\varepsilon}$ itself. In particular, for relaxation parameter $\varepsilon \in (0,1)$ and fixed dimension $m,N \in \N$, $m < N$ any (single-valued) map in $\mathcal{M}_{\Xi_{\text{bp},m,N,\varepsilon}}$ is not computable on Turing machines. These maps represent the practically relevant task to compute exactly one suitable reconstruction. Does this result still hold when considering the BSS model? In the BSS model the situation is in fact quite different and a general non-computability result does not hold any more. Under certain circumstances we can even guarantee the computability of at least one function in $\mathcal{M}_{\Xi_{\text{bp},m,N,\varepsilon}}$.

We now consider the real and complex case separately, since the BSS model shows different behaviour depending on the underlying structure.

\subsection{Real Case}
\noindent
Given a multi-valued mapping $f:\mathcal{X} \rightrightarrows \mathcal{Y}$ we denote by $f^\R$ its restriction to real inputs and outputs. Although only the complex domain was considered explicitly in \cite{Boche2022LimitsDL}, the non-computability in the Turing model of any map in $\mathcal{M}_{\Xi^\R_{\text{bp},m,N,\varepsilon}}$ is still valid when restricting to the real domain. This follows immediately from the fact that the given proof is also accurate in the strictly real case. In contrast, we show that there exists at least one element in $\mathcal{M}_{\Xi^\R_{\text{bp},m,N,\varepsilon}}$ that is computable in BSS sense. Even more, we prove that the dependence of the relaxation parameter $\varepsilon$ can be incorporated in the reconstruction map without negatively influencing the computability property. Hence, there exists a function in $\mathcal{M}_{\Xi^\R_{\text{bp},m,N}}$ that is BSS-computable.    

\begin{Lemma}\label{lemma:BSSReal}
    For $A \in \R^{m \times N}$, $y \in \R^m$, and $\varepsilon>0$ set $S_{(A,y,\varepsilon)} \coloneqq \{ x \in \R^N \mid \norm[\ell_2]{Ax -y} \leq \varepsilon\}$ and consider the multi-valued map
    \begin{align*} 
        \Psi : \R^{m\times N} \times \R^m \times \R_{\geq 0}  &\rightrightarrows \R^N \\
        (A,y,\varepsilon) &\mapsto \{x^\ast \in \R^N \mid  \norm[\ell_1]{x^\ast} = \min_{x \in S_{(A,y,\varepsilon)}} \norm[\ell_1]{x} \text{ and } \norm[\ell_2]{Ax^\ast -y} \leq \varepsilon \}.
    \end{align*}    
    Then, there exists a (single-valued) function $\Psi^s \in \mathcal{M}_\Psi$ that is BSS-computable.
\end{Lemma}

\begin{proof}
We need to show that there exists a program for a BSS machine $\mathcal{B}$ so that its input-output map $\Phi_\mathcal{B}$ is equivalent to a function $\Psi^s \in \mathcal{M}_\Psi$. Given $A \in \R^{m \times N}$, $y \in \R^m$, and $\varepsilon > 0$ we therefore have to compute $x^\ast$ with   
\begin{equation}\label{eq:RealInvProb}
    x^\ast \in \argmin_{x \in \R^N} \norm[\ell_1]{x} \text{ such that } \norm[\ell_2]{Ax -y} \leq \varepsilon.
\end{equation}
For this, let $A = (a_{i,j})$ and note that
\begin{equation*}
    \norm[\ell_2]{Ax -y} \leq \varepsilon \quad \Leftrightarrow \quad \norm[\ell_2]{Ax -y}^2 \leq \varepsilon^2 \quad \Leftrightarrow \quad \sum_{i=1}^m\big(\sum_{j=1}^N a_{i,j}x_j -y_i \big)^2 -\varepsilon^2 \leq 0.     
\end{equation*}
This implies
\begin{equation*}
    S_{(A,y,\varepsilon)} = \{x\in\R^N \mid \norm[\ell_2]{Ax -y} \leq \varepsilon \} = \{x\in\R^N \mid q(x) \leq 0 \}    
\end{equation*}
where $q: \R^N \to \R$ is a polynomial, i.e., $S_{(A,y,\varepsilon)}$ is a semialgebraic set with a quantifier free description.

Now, our goal is to apply the optimization algorithm, \Cref{alg:GlOp}, to find a minimizer of a polynomial on $S_{(A,y,\varepsilon)}$. In its current form, $\norm[\ell_1]{x} = \sum_i \abs{x_i}$ is not a polynomial, however, we can rewrite it as
\begin{equation*}
    \norm[\ell_1]{x} = \sum_i \abs{x_i} = \sum_i x_i^+ + x_i^-
\end{equation*}
where $x_i^+ = \max\{ 0, x_i\}$ and $x_i^- = - \min\{ 0,x_i\}$. Observe that $p: \R^{2N} \to \R$, $(x^+,x^-) \mapsto \sum_i {x_i^+ + x_i^-}$ is indeed a polynomial.

Hence, we can conclude that \eqref{eq:RealInvProb} is equivalent to the problem
\begin{equation*}
    (w^\ast,z^\ast)\in \argmin_{\substack{w,z \in \R^N,\\ w,z \geq 0}} \sum_{i=1}^N w_i+z_i \text{ such that } \norm[\ell_2]{A(w-z) -y} \leq \varepsilon.
\end{equation*}
Now it is easy to deduce that the task is to minimize the polynomial $f: \R^{2N} \to \R$, $f(w,z)= \sum_{i=1}^N w_i+z_i$ on 
\begin{align}\label{eq:pols}
    &\{(w,z)\in\R^{2N} \mid \norm[\ell_2]{A(w-z) -y} \leq \varepsilon, w\geq 0, z\geq 0 \}\nonumber \\
    &=\{(w,z)\in\R^{2N} \mid q(w,z) \leq 0 \wedge \bigwedge_{i=1}^N  r_i(w,z) \geq 0 \wedge \bigwedge_{i=1}^N s_i(w,z) \geq 0, \} 
\end{align} 
where $q,r_i,s_i: \R^{2N} \to \R$ are polynomials --- $q$ describes the $\ell_2$ norm term whereas $r_i(w,z) = w_i, s_i(w,z)= z_i$ describe the non-negativity conditions. Consequently, 
\begin{equation*}
    \{(w,z)\in\R^{2N} \mid \norm[\ell_2]{A(w-z) -y} \leq \varepsilon, w\geq 0, z\geq 0 \}
\end{equation*}
is a semialgebraic set with a quantifier free description. Thus, we can apply the optimization algorithm, \Cref{alg:GlOp}, to minimize the polynomial $f$ on a semialgebraic set and obtain a corresponding minimizer $(w^\ast,z^\ast) \in\R^{2N}$. By construction, the difference $x^\ast = w^\ast - z^\ast$ is a minimizer of the original problem \eqref{eq:RealInvProb}. Thus, we can define a program for $\mathcal{B}$ which, on input $A \in \R^{m \times N}$, $y \in \R^m$, and $\varepsilon > 0$, yields the polynomials $q,r_i,s_i$ specified in \eqref{eq:pols}, uses \Cref{alg:GlOp} to obtain a minimizer $(w^\ast,z^\ast)$ of $f$, and finally outputs the solution $x^\ast = w^\ast - z^\ast$. This completes the proof.
\end{proof}

We immediately obtain the following theorem by applying \Cref{lemma:BSSReal}.
\begin{Theorem}\label{thm:BSSReal}
    There exists a BSS-computable function $\Xi^s \in \mathcal{M}_{\Xi^\R_{\text{bp},m,N}}$.
\end{Theorem}
\begin{Remark}
    We can conclude that inverse problems described by $\Xi^\R_{\text{bp},m,N}$ are algorithmically solvable in the BSS model. 
    The special case for fixed $\varepsilon>0$ is included in \Cref{thm:BSSReal}, i.e., there exists at least one BSS-computable function in $\mathcal{M}_{\Xi^\R_{\text{bp},m,N, \varepsilon}}$. In contrast, we already know that every function in $\mathcal{M}_{\Xi^\R_{\text{bp},m,N, \varepsilon}}$ is not computable on a Turing machine for $\varepsilon \in (0,1)$ \cite{Boche2022LimitsDL}. Hence, inverse problems %described by $\text{bp}^\R_{m,N,\varepsilon}$ are algorithmically solvable on BSS machines but not on Turing machines.
    via basis pursuit minimization \eqref{eq:sparseprob} are algorithmically solvable on BSS machines but not on Turing machines.
\end{Remark}
\Cref{thm:BSSReal} implies that there exists an algorithm (in the BSS sense), which for every sampling operator $A$ and measurement $y$ yields a feasible reconstruction $x$. However, we did not specify which reconstruction (of all feasible reconstruction) is found by the algorithm. This is not a problem in practice as all feasible reconstructions are equally valid.

Now, we pose the question whether the algorithmic solvability of inverse problems in the BSS model is connected to specific properties of basis pursuit minimization, i.e., is algorithmic solvability valid for inverse problems with different regularization schemes as well? For this, we consider lasso minimization \eqref{eq:lasso}. It turns out that algorithmic solvability is in fact also achievable in the lasso formulation.

\begin{Lemma}\label{lemma:BSSReallasso}
    For $A \in \R^{m \times N}$, $y \in \R^m$, and $\lambda>0$ consider the multi-valued map
    \begin{align*} 
        \Psi : \R^{m\times N} \times \R^m \times \R_{\geq 0}  &\rightrightarrows \R^N \\
        (A,y,\varepsilon) &\mapsto \{x^\ast \in \R^N \mid  \norm[\ell_1]{x^\ast} = \min_{x \in \R^N} \lambda \norm[\ell_1]{x} +  \norm[\ell_2]{Ax-y}^2 \}.
    \end{align*}    
    Then, there exists a (single-valued) function $\Psi^s \in \mathcal{M}_\Psi$ that is BSS-computable.
\end{Lemma}
\begin{proof}
We need to show that there exists a program for a BSS machine $\mathcal{B}$ so that its input-output map $\Phi_\mathcal{B}$ is equivalent to a function $\Psi^s \in \mathcal{M}_\Psi$. Given $A \in \R^{m \times N}$, $y \in \R^m$, and $\varepsilon > 0$ we therefore have to compute $x^\ast$ with   
\begin{equation}\label{eq:RealInvProblasso}
    x^\ast \in \argmin_{x \in \R^N} \lambda \norm[\ell_1]{x} +  \norm[\ell_2]{Ax-y}^2.
\end{equation}
For this, let $A = (a_{i,j})$ and note that
\begin{align*}
    \lambda \norm[\ell_1]{x} +  \norm[\ell_2]{Ax-y}^2 &= \lambda \sum_i \abs{x_i} + \sum_{i=1}^m\big(\sum_{j=1}^N a_{i,j}x_j -y_i \big)^2    \\
    &= \lambda \sum_i x_i^+ + x_i^- + \sum_{i=1}^m\big(\sum_{j=1}^N a_{i,j}(\sum_i x_i^+ - x_i^-) -y_i \big)^2
\end{align*}
where $x_i^+ = \max\{ 0, x_i\}$ and $x_i^- = - \min\{ 0,x_i\}$. This implies that a solution of \eqref{eq:RealInvProblasso} can be obtained via
\begin{equation}\label{eq:optlassosol}
   (w^\ast,z^\ast)\in \argmin_{\substack{w,z \in \R^N,\\ w,z \geq 0}} \lambda \sum_i w_i^\ast + z_i^\ast + \sum_{i=1}^m\big(\sum_{j=1}^N a_{i,j}(\sum_i w_i^\ast - z_i^\ast) -y_i \big)^2.   
\end{equation}
as $x^\ast =  w^\ast - z^\ast$. Now, consider the polynomial $f: \R^{2N} \to \R$ given by
\begin{equation*}
     f(w,z) = \lambda \sum_i w_i + z_i + \sum_{i=1}^m\big(\sum_{j=1}^N a_{i,j}(\sum_i w_i - z_i) -y_i \big)^2
\end{equation*}
and the semialgebraic set with a quantifier free description
\begin{align*}
    &\{(w,z)\in\R^{2N} \mid \, w\geq 0, z\geq 0 \}\nonumber \\
    &=\{(w,z)\in\R^{2N} \bigwedge_{i=1}^N  r_i(w,z) \geq 0 \wedge \bigwedge_{i=1}^N s_i(w,z) \geq 0 \}, 
\end{align*} 
where $r_i,s_i: \R^{2N} \to \R$, $r_i(w,z) = w_i, s_i(w,z)= z_i$ are polynomials describing the non-negativity conditions. We can apply the optimization algorithm, \Cref{alg:GlOp}, to find a minimizer $(w^\ast,z^\ast) \in\R^{2N}$ of $f$ on the semialgebraic set specified by the polynomials $r_i,s_i$. By construction $(w^\ast,z^\ast)$ is a solution of \eqref{eq:optlassosol} and the difference $x^\ast =  w^\ast - z^\ast$ is a sought solution of \eqref{eq:RealInvProblasso}. Consequently, there exists a program for $\mathcal{B}$ which, on input $A \in \R^{m \times N}$, $y \in \R^m$, and $\lambda > 0$, yields the polynomials $r_i,s_i$ and uses \Cref{alg:GlOp} to obtain a minimizer $(w^\ast,z^\ast)$ of $f$, and finally outputs the solution $x^\ast$. This completes the proof.
\end{proof}
Now, the computability in BSS sense of lasso minimization is immediate via \Cref{lemma:BSSReallasso}.
\begin{Theorem}\label{thm:BSSReallasso}
    There exists a BSS-computable function $\Xi^s \in \mathcal{M}_{\Xi^\R_{\text{la},m,N}}$.
\end{Theorem}
\begin{Remark}
    We can infer that inverse problems described by $\Xi^\R_{\text{la},m,N}$ are algorithmically solvable in the BSS model. Again, the special case for fixed optimization parameter $\lambda>0$ is also covered by the statement of \Cref{thm:BSSReallasso}. In \cite{Boche2022LimitsDL}, it is shown that inverse problems via square root lasso optimization are not algorithmically solvable on Turing machines. Square root lasso optimization is defined as 
    \begin{equation}\label{eq:srlasso}
        \argmin_{x \in \C^N} \lambda \norm[\ell_1]{x} + \norm[\ell_2]{Ax -y}, \tag{sla}
    \end{equation}
    i.e., the only difference to lasso \eqref{eq:lasso} is the power of the $\ell_2$ norm. Therefore, the analogous proof technique can be applied to derive algorithmic non-solvability of lasso minimization. Thus, we also obtain different capabilities of Turing and BSS machines for solving inverse problems via lasso minimization. However, square root lasso optimization can not be directly tackled on BSS machines due to the non-computability of the square root function connected to the $\ell_2$ norm. Hence, similar modifications as described in \Cref{sec:complex} for basis pursuit and lasso need to be applied to circumvent this problem.    
\end{Remark}

In summary, there exists a clear distinction between algorithmic solvability of inverse problems restricted to the real domain on Turing and BSS machines. Whereas Turing machines do not allow for algorithmic solvability of inverse problems \cite{Boche2022LimitsDL}, in the BSS model we can prove algorithmic solvability for various underlying optimization formulations as basis pursuit and lasso. Hence, the limitations arising in the Turing model, do not generally occur in the BSS model for solving inverse problems algorithmically.

\subsection{Complex Case}\label{sec:complex}
\noindent
The study of computability is more evolved in the complex domain. The reason being that BSS machines over the complex field have different properties than BSS machines over the real numbers. Complex BSS machines treat complex numbers as entities whereas real BSS machines rely on real numbers as basic entities. One main difference arises from the fact that complex numbers are not an ordered field. Thus, complex BSS machines can not compare arbitrary complex numbers but only check the equality to zero at their branch nodes. Moreover, even basic functions connected to complex numbers such as $z \mapsto \Re(z)$, $z \mapsto \Im(z)$, $z \mapsto \bar{z}$, and $z \mapsto \abs{z}$ are not BSS-computable which follows from the fact that $\R$ is not a BSS-decidable set in $\C$ \cite{Blum98ComplRealComp}. Therefore, $\ell_p$ norms are not BSS-computable on complex inputs as they require to calculate the absolute value of the inputs. A possible work-around is to represent $\C$ by $\R^2$ and apply operations on real numbers. That is, we basically consider the real model and employ real BSS machines that take complex inputs $x$ in form of $(\Re(x), \Im(x))$. Then, the functions $z \mapsto \Re(z)$ and $z \mapsto \Im(z)$ are computable on a real BSS machine since it only requires the machine to process the respective part of the representation of $z$. Hence, the representation of the complex field influences the capabilities of the corresponding BSS machine. Unfortunately, this is still not sufficient to guarantee computability of the the $\ell_1$ norm
\begin{equation*}
    \norm[\ell_1]{x} = \sum_i \abs{x_i} = \sum_i \sqrt{\Re(x_i)^2 + \Im(x_i)^2}
\end{equation*}
on real BSS machines, since the square root function $\sqrt{\cdot}$ is not BSS-computable on $\R$. % and consequently the $\ell_1$ normis not BSS-computable on a real BSS machine. 
We wish to remark that the non-computability of the square root function on BSS machines assumes exact computations, which we consider throughout this work. By relaxing to approximate computations the non-computability on BSS machines may be avoided. Similarly, on Turing machines the square root function can also not be computed exactly, even on a discrete set such as $\N$ since the irrational number $\sqrt{2}$ can not be represented exactly on a Turing machine. However, the square root function can be approximated to arbitrary degree on Turing machines, i.e., it is a Turing-computable function on $\R$. 

Nevertheless, we can observe that slight modifications suffice to obtain a real BSS-computable function. Replacing the $\ell_1$ by the squared $\ell_2$ norm yields a BSS-computable function, since

\begin{equation*}
    \norm[\ell_2]{x}^2 = \sum_i \abs{x_i}^2 = \sum_i \Re(x_i)^2 + \Im(x_i)^2
\end{equation*}
can be expressed by elementary field operations and the real BSS-computable functions $\Re$ and $\Im$. This is in fact no exception, similar arguments can be made for any $\ell_p$ norm when $p$ is even. 

In general, the basis pursuit formulation in \eqref{eq:sparseprob} promotes sparse solutions due to the $\ell_1$ norm in the objective function. This property is neglected when replacing the $\ell_1$ with the squared $\ell_2$ norm in the objective. Hence, a potentially BSS-computable substitute in the objective ideally remains close to the original objective function so that desired properties are maintained. Therefore, we introduce a norm $\norm[\ast]{\cdot}$ which satisfies these requirements and at the same time emphasizes the difference between the capabilities of BSS machines and Turing machines. We define $\norm[\ast]{\cdot}$ on $\C^N$ by
\begin{equation*}
    \norm[\ast]{x} \coloneqq \sum_i^N \abs{\Re(x_i)} + \abs{\Im(x_i)}
\end{equation*}
Next, we show that by changing the objective in basis pursuit \eqref{eq:sparseprob} from $\ell_1$ norm to $\norm[\ast]{\cdot}$ we obtain an algorithmically solvable optimization problem.
%and show that, by changing the objective in basis pursuit \eqref{eq:sparseprob} from $\ell_1$ norm to $\norm[\ast]{\cdot}$, the  we obtain a real BSS-computable function.

\begin{Lemma}\label{lemma:BSSComplex}
    For $A \in \C^{m \times N}$, $y \in \C^m$, and $\varepsilon>0$ set $S_{(A,y,\varepsilon)} \coloneqq \{ x \in \C^N \mid \norm[\ell_2]{Ax -y} \leq \varepsilon\}$ and consider the multi-valued map
    \begin{align*} 
        \Psi : \C^{m\times N} \times \C^m \times \R_{\geq 0}  &\rightrightarrows \C^N \\
        (A,y,\varepsilon) &\mapsto \{x^\ast \in \C^N \mid  \norm[\ast]{x^\ast} = \min_{x \in S_{(A,y,\varepsilon)}} \norm[\ast]{x} \text{ and } \norm[\ell_2]{Ax^\ast -y} \leq \varepsilon \}.
    \end{align*}    
    Then, there exists a (single-valued) function $\Psi^s \in \mathcal{M}_\Psi$ that is BSS-computable.
\end{Lemma}

\begin{Remark}
    We employ BSS machines over the real numbers by representing complex inputs via their real and imaginary parts as explained in the previous paragraph. For other representations of complex numbers and the corresponding BSS machines, e.g., complex BSS machines, the statement does not hold.  %Additionally, like in the real case the choice of the minimizer $x^\ast$ for each input is not prescribed. There has to exist at least one for each input so that the overall mapping $\Psi$ is computable for this specific choice.     
\end{Remark}

\begin{proof}
We need to show that there exists a program for a real BSS machine $\mathcal{B}$ so that its input-output map $\Phi_\mathcal{B}$ is equivalent to a function $\Psi^s \in \mathcal{M}_\Psi$. Given $A \in \C^{m \times N}$, $y \in \C^m$, and $\varepsilon > 0$, whereby each complex element is represented by its real and imaginary part, we have to compute $x^\ast$ with   
\begin{equation}\label{eq:ComplInvProb}
    x^\ast \in \argmin_{x \in \C^N} \norm[\ast]{x} \text{ such that } \norm[\ell_2]{Ax -y} \leq \varepsilon.
\end{equation}
Let $A = (a_{i,j})$ and note that
\begin{equation}\label{eq:obs1}
    \norm[\ell_2]{Ax -y} \leq \varepsilon \quad \Leftrightarrow \quad \norm[\ell_2]{Ax -y}^2 \leq \varepsilon^2 \quad \Leftrightarrow \quad \sum_{i=1}^m\left(\abs{\sum_{j=1}^N a_{i,j}x_j -y_i}\right)^2 -\varepsilon^2 \leq 0  
\end{equation}
and for fixed $i\in \{1,\dots,m\}$
\begin{align}\label{eq:obs2}
    \left(\abs{\sum_{j=1}^N a_{i,j}x_j -y_i}\right)^2 &= \left(\Re\left(\sum_{j=1}^N a_{i,j}x_j -y_i\right)\right)^2 + \left(\Im\left(\sum_{j=1}^N a_{i,j}x_j -y_i\right)\right)^2\nonumber\\
    &= \left(\sum_{j=1}^N \Re(a_{i,j})\Re(x_j) - \Im(a_{i,j})\Im(x_j) - \Re(y_i) \right)^2 + \nonumber\\
    &\phantom{vvvvvvvvvv} \left(\sum_{j=1}^N \Re(a_{i,j})\Im(x_j) + \Im(a_{i,j})\Re(x_j) - \Im(y_i) \right)^2.
\end{align}
Combining \eqref{eq:obs1} and \eqref{eq:obs2} yields
\begin{align*}
    \norm[\ell_2]{Ax -y} \leq \varepsilon \quad \Leftrightarrow \quad \sum_{i=1}^m &\left(\sum_{j=1}^N \Re(a_{i,j})\Re(x_j) - \Im(a_{i,j})\Im(x_j) - \Re(y_i) \right)^2\\
    &+ \left(\sum_{j=1}^N \Re(a_{i,j})\Im(x_j) + \Im(a_{i,j})\Re(x_j) - \Im(y_i) \right)^2 -\varepsilon^2 \leq 0,     
\end{align*}
i.e., we can identify the set $\{x \in \C^N | \norm[\ell_2]{Ax -y} \leq \varepsilon \}$ with $\{ x \in \R^{2N} | q(x) \leq 0\}$ where $q: \R^{2N} \to \R$ is a polynomial. Hence, $\{x \in \C^N | \norm[\ell_2]{Ax -y} \leq \varepsilon \}$ can be represented by a (real) semialgebraic set with a quantifier free description. In order to use the optimization algorithm, \Cref{alg:GlOp}, it remains to show that the objective function in \eqref{eq:ComplInvProb} is a real polynomial $f: \R^{2N} \to \R$.

However, the objective function $\norm[\ast]{\cdot}$ is not a polynomial (due to the absolute values in its definition). By applying the same workaround as in the proof of \Cref{lemma:BSSReal} --- rewriting the absolute value as the sum of two variables --- we obtain that 
\begin{equation*}
    \norm[\ast]{x} = \sum_i  \abs{\Re(x_i)} + \abs{\Im(x_i)} = \sum_i  \Re(x_i)^+ + \Re(x_i)^- + \Im(x_i)^+ + \Im(x_i)^-,
\end{equation*}
where $x_i^+ = \max\{ 0, x_i\}$ and $x_i^- = - \min\{ 0,x_i\}$. Thus, $\norm[\ast]{\cdot}$ can be represented by a polynomial
$p: \R^{4N} \to \R$,
\begin{equation*}
    p(\Re(x_i)^+,\Re(x_i)^-,\Im(x_i)^+,\Im(x_i)^-) = \sum_i {\Re(x_i)^+ + \Re(x_i)^- + \Im(x_i)^+ + \Im(x_i)^-}.    
\end{equation*}
Therefore, the same reasoning as in the real case shows that a minimizer $\hat{x} \in \R^{2N}$ of $\norm[\ast]{\cdot}$ on $\{ x \in \R^{2N} | q(x) \leq 0\}$ can be algorithmically computed. The sought minimizer $x^\ast \in \C^N$ of \eqref{eq:ComplInvProb} is then simply $\hat{x}$, whereby the elements in $\hat{x}$ are considered as the real and imaginary parts of $x^\ast$. 
\end{proof}

We immediately obtain the following theorem by applying \Cref{lemma:BSSComplex}.
\begin{Theorem}\label{thm:BSSComplex}
    For $A\in\C^{m\times n}$, $y\in \C^m$ and $\varepsilon >0$ consider the optimization problem 
    \begin{equation}\label{eq:adjProb}
        \argmin_{x \in \C^N} \norm[\ast]{x} \text{ such that } \norm[\ell_2]{Ax -y} \leq \varepsilon. \tag{$\text{bp}^\ast$}
    \end{equation}
    Then, there exists a BSS-computable function $\Xi^s \in \mathcal{M}_{\Xi_{\text{bp}^\ast,m,N}}$.   
\end{Theorem}
\begin{Remark}
    Again, the special case for fixed $\varepsilon>0$ is included in \Cref{thm:BSSComplex}, i.e., there also exists a BSS-computable function $\Xi^s \in \mathcal{M}_{\Xi_{\text{bp}^\ast,m,N,\varepsilon}}$. 
\end{Remark}
With the same adjustment, i.e., replacing the $\norm[\ell_1]{\cdot}$ by $\norm[\ast]{\cdot}$, and the analogous proof technique we obtain a similar result for lasso optimization \eqref{eq:lasso}.
\begin{Theorem}\label{thm:BSSComplexlasso}
    For $A\in\C^{m\times n}$, $y\in \C^m$ and $\lambda >0$ consider the optimization problem 
    \begin{equation}\label{eq:adjProblasso}
        \argmin_{x \in \C^N} \lambda \norm[\ast]{x} + \norm[\ell_2]{Ax -y}^2. \tag{$\text{la}^\ast$}
    \end{equation}
    Then, there exists a BSS-computable function $\Xi^s \in \mathcal{M}_{\Xi_{\text{la}^\ast,m,N}}$.  
\end{Theorem}
\begin{Remark}
    We can conclude that inverse problems expressed through the optimization problems \eqref{eq:adjProb} and \eqref{eq:adjProblasso} with the modified objective $\norm[\ast]{\cdot}$ are algorithmically solvable in the BSS model. Hence, there exists algorithms in the BSS sense that for every sampling operator $A$ and measurement $y$ yield a feasible reconstruction $x \in \Xi_{\text{bp}^\ast,m,N}$ and $x \in \Xi_{\text{la}^\ast,m,N}$, respectively.
\end{Remark}
Several questions concerning our results arise. In particular, it is desirable to relate our findings to the algorithmic solvability of inverse problems on Turing machines.
\begin{itemize}
    \item First, how do the modifications of the optimization problems influence the outcomes in the Turing model? As already mentioned, basis pursuit and lasso optimization are not algorithmically solvable on Turing machines \cite{Boche2022LimitsDL}. Applying the same proof technique also yields algorithmic non-solvability in the Turing model for the adjusted description in \eqref{eq:adjProb} and \eqref{eq:adjProblasso}. We demonstrate the proof technique in \Cref{sec:app} for \eqref{eq:adjProb}.
    %In particular, it is not hard to show that they also hold for the description in \eqref{eq:adjProb} (see \Cref{sec:app} for details) so that algorithmic non-solvability in the Turing model immediately follows.
    \item Second, can we find similar modifications of the optimization problems which allow to solve inverse problems successfully on a Turing machine? Here, the answer appears to be negative. The reasoning for algorithmic non-solvability of basis pursuit and lasso in the Turing model was in principle not connected to the utilized objective function. Rather, non-computability conditions were established that are fairly independent of the objective or, more specifically, can be easily adapted to certain changes in the objective function. The requirements to apply the non-computability conditions are mainly connected to the properties of the solution set of a given problem. If the solution set satisfies certain properties, then the non-computability statement can also be extended to the corresponding problem formulation. Hence, as long as the modification of the objective function does not promote substantial changes in the solution set the non-computability statement remains valid. Moreover, changing the objective entirely may result in a quality loss in the computed solutions since the desired properties such as sparsity may be no longer encourages by the new objective.
    \item Therefore, our final questions directly relates to the previous discussion: Can we justify the changes in the objective function resulting in the optimization problems \eqref{eq:adjProb} and \eqref{eq:adjProblasso}? In particular, how do  \eqref{eq:adjProb} and \eqref{eq:adjProblasso} relate to the original problems in \eqref{eq:sparseprob} and \eqref{eq:lasso}, respectively? Although these adjustments would typically not be applied to solve inverse problems in practice, the approach tries to maintain the underlying properties of the original formulations. Hence, it is a striking result that there exists fairly similar descriptions of the inverse problem setting that can be successfully solved in the BSS model. 
\end{itemize}
Next, we want to nevertheless establish the modifications of the optimization problems more rigorously without introducing an appropriate but to a certain degree arbitrary norm. Therefore, our goal is to approximate the $\ell_1$ norm adequately instead of replacing it. We exemplify this approach on the basis pursuit problem. The proof is provided in \Cref{sec:appb}.
\begin{Theorem}\label{thm:BPapprox}
    Let $\beta, \gamma > 0$. For $A\in\C^{m\times N}$, $y\in \C^m$ and $\varepsilon >0$ consider the optimization problem 
    \begin{equation}\label{eq:approxbp}
        \argmin_{x \in \C^N} p_{\beta, \gamma}(x)  \text{ such that }   \norm[\ell_2]{Ax -y} \leq \varepsilon \,\, \wedge \,\, \norm[\ell_2]{x} < \sqrt{N} \beta,  \tag{$\text{p}(\beta, \gamma)$}
    \end{equation}
    where $p_{\beta, \gamma}$ is a polynomial satisfying 
    \begin{equation}\label{eq:approxl1}
        \sup_{x \in \C^N:  \norm[\ell_2]{x} < \sqrt{N} \beta}\abs{\norm[\ell_1]{x} - p_{\beta, \gamma}(x)} \leq \gamma.
    \end{equation}
    Then, there exists a BSS-computable function $\Xi^s \in \mathcal{M}_{\Xi_{\text{p}(\beta,\gamma),m,N}}$. Moreover, there also exists a BSS-computable function $g: \C^{m\times N} \times \C^m \times \R_{>0} \times \R_{>0}\to \{0,1\}$ so that $g(A,y, \varepsilon,\beta) = 1$ is equivalent to the following statement: For $(A,y, \varepsilon)$ there exists at least one solution of basis pursuit \eqref{eq:sparseprob} and the solution(s) are contained in $I_\beta \coloneqq \{x \in \C^N: \norm[\ell_2]{x} < \sqrt{N} \beta\}$.    
    %implies that the solutions of basis pursuit \eqref{eq:sparseprob} are contained in $I_\beta \coloneqq \{x \in \C^N: \norm[\ell_2]{x} < \sqrt{N} \beta\}$.
\end{Theorem}
\begin{Remark}
    We can conclude that inverse problems described by $\Xi_{\text{p}(\beta,\gamma),m,N}$ are algorithmically solvable in the BSS model. Note that the objective $p_{\beta, \gamma}$ in \eqref{eq:approxbp} approximates up to an error of $\gamma$ the $\ell_1$ norm, i.e., the objective of the original basis pursuit optimization, on the set $I_\beta$. Therefore, \eqref{eq:approxbp} represents an approximation of basis pursuit if its minimizers are contained in $I_\beta$. This property can be checked via the function $g$ in \Cref{thm:BPapprox}, i.e., before evaluating $\Xi^s \in \mathcal{M}_{\Xi_{\text{p}(\beta,\gamma),m,N}}$ for given $(A,y,\varepsilon)$ we can algorithmically verify if $\Xi_{\text{bp},m,N}(A,y,\varepsilon) \subset I_\beta$. However, the minimizers of \eqref{eq:approxbp} and basis pursuit need not to agree and we do not obtain worst-case bounds on their distance.   
\end{Remark}
\begin{Remark}
    The algorithmic solvability of $\Xi_{\text{p}(\beta,\gamma),m,N}$ and the existence of $g$ imply that there exists a BSS machine $\mathcal{B}_{\beta,\gamma}$ that computes a solution of \eqref{eq:approxbp} given $A\in\C^{m\times n}$, $y\in \C^m$ and $\varepsilon>0$ as input, provided that $(A,y,\varepsilon)$ satisfy certain properties. In particular, $\mathcal{B}_{\beta,\gamma}$ first checks if the solutions of basis pursuit for $(A,y,\varepsilon)$ are contained in $I_\beta$. In this way, we can ensure that the approximation of the the $\ell_1$ norm by $p_{\beta, \gamma}$ is valid on the relevant domain. Otherwise, the computation is aborted. If not, a solution of \eqref{eq:approxbp} for $(A,y,\varepsilon)$ is computed. Note that $\mathcal{B}_{\beta,\gamma}$ needs to be constructed for fixed approximation accuracy $\gamma$ and fixed acceptance domain depending on $I_\beta$, i.e., $\beta$ and $\gamma$ are not part of the input to the machine. Therefore, a change of the acceptance domain and approximation accuracy entails a new construction of the associated BSS machines $\mathcal{B}_{\beta,\gamma}$. In particular, we need to encode a finite set of constants, which depend on $\beta$ and $\gamma$ but can not be computed by a BSS machine directly. The constants enable $\mathcal{B}_{\beta,\gamma}$ to compute an appropriate polynomial $p_{\beta,\gamma}$ and subsequently a minimizer of \eqref{eq:approxbp}.
\end{Remark}
\begin{Remark}
    In the Turing model, the outlined approach to approximate basis pursuit is not feasible; we refer to \Cref{sec:appb} for details. Nevertheless, different approximation schemes may exist Turing machines. However, we have lower bounds on how accurate an algorithm executed on a Turing machine can be \cite{Boche2022LimitsDL}, i.e., how close an algorithm on a Turing machine can approximate the minimizers of basis pursuit.
\end{Remark}

\section{Discussion}\label{sec:Discussion}
\noindent
Our findings have noteworthy implications concerning the solvers of finite-dimensional inverse problems. In particular, they allow to characterize the boundaries of any general algorithm applied to solve these inverse problems. More precisely, we examined the limits of any algorithm implemented on a specific computing devices imposed by the underlying hardware, in particular, any solver implemented on the given computing device is subject to these restrictions.
On digital hardware, which is the currently predominant hardware platform, the algorithmic solvability via basis pursuit \eqref{eq:sparseprob} and lasso \eqref{eq:lasso} is limited \cite{Boche2022LimitsDL}. On analog hardware, represented by the mathematical model of BSS machines, our findings are more favourable. First, algorithmic solvability on BSS machines of inverse problems can be established when restricting from the complex to the real domain. Furthermore, we can find modifications \eqref{eq:adjProb}, \eqref{eq:adjProblasso}, and \eqref{eq:approxbp} that allow for algorithmic solvability on BSS machines while staying `close' to the original basis pursuit and lasso problem. At the same time, the algorithmic non-solvability on Turing machines, i.e., digital computers, still persists.  Hence, the utilized computing device has tremendous influence on the existence and capabilities of solvers applied to inverse problems. 

What are the effects of these results on current solution techniques? First, today's solvers are in general implemented on digital hardware. Thus, they are subject to the limitations of digital hardware and even improving the solvers will not suffice to circumvent the imposed boundaries. Therefore, DL on digital hardware as the premium approach to solve inverse problems is subject to these limitations. This imposes restrictions on the reliability and trustworthiness of DL on digital hardware for solving inverse problems. In particular, a general and reliable DL solver for inverse problems can not exist if formal proof of correctness or a certificate for the output is expected.  

Second, future developments in hardware technology (as described in \Cref{seq:AnaDigComp}) to analog hardware may enable more powerful solvers. Here, the question is whether the BSS model is an appropriate model to describe the capabilities of the upcoming analog hardware solution. Although a BSS machine is certainly an idealized model, due to exact real number processing, an important point is whether and to what degree these models can be implemented and approximated. In particular, the effects of noisy computations need to be evaluated. For now, we can conclude that in an idealized theoretical setting analog hardware (modeled by BSS machines) offers a platform for potentially reliable DL.   

\section*{Acknowledgements}
This work of H. Boche was supported in part by the German Federal Ministry of Education and Research (BMBF) in the project Hardware Platforms and Computing Models for Neuromorphic Computing (NeuroCM) under Grant 16ME0442 and within the national initiative on 6G Communication Systems through the research hub 6G-life under Grant 16KISK002.

G. Kutyniok acknowledges support from LMUexcellent, funded by the Federal Ministry of Education and Research (BMBF) and the Free State of Bavaria under the Excellence Strategy of the Federal Government and the Länder as well as by the Hightech Agenda Bavaria. Further, G. Kutyniok was supported in part by the DAAD programme Konrad Zuse Schools of Excellence in Artificial Intelligence, sponsored by the Federal Ministry of Education and Research. G. Kutyniok also acknowledges support from the Munich Center for Machine Learning (MCML) as well as the German Research Foundation under Grants DFG-SPP-2298, KU 1446/31-1 and KU 1446/32-1 and under Grant DFG-SFB/TR 109, Project C09 and the German Federal Ministry of Education and Research (BMBF) under Grant MaGriDo.

\printbibliography

@incollection{Berner2021modernMathDL,
      title={The {M}odern {M}athematics of {D}eep {L}earning}, 
      booktitle = {Mathematical Aspects of Deep Learning},
      author={Julius Berner and Philipp Grohs and Gitta Kutyniok and Philipp Petersen},
      year={2022},
      publisher={Cambridge University Press}
}

@ARTICLE{Boche2022LimitsDL,
  author={Boche, Holger and Fono, Adalbert and Kutyniok, Gitta},
  journal={IEEE Trans. Inf. Theory }, 
  title={Limitations of Deep Learning for Inverse Problems on Digital Hardware}, 
  year={2023},
  volume={69},
  number={12},
  pages={7887-7908},
}

@ARTICLE{Grozinger2019Biocomp,
    author = {{Grozinger}, Lewis and {Amos}, Martyn and {Gorochowski}, Thomas E. and {Carbonell}, Pablo and {Oyarz{\'u}n}, Diego A. and {Stoof}, Ruud and {Fellermann}, Harold and {Zuliani}, Paolo and {Tas}, Huseyin and {Go{\~n}i-Moreno}, Angel},
    title = "{Pathways to cellular supremacy in biocomputing}",
    journal = {Nat. Commun.},
    year = {2019},
    volume = {10},
}

@article{Blum89BSSmachines,
    author = {Lenore Blum and Mike Shub and Steve Smale},
    title = {On a theory of computation and complexity over the real numbers: {NP}-completeness, recursive functions and universal machines},
    volume = {21},
    journal = {Bull. (New Ser.) Am. Math. Soc.},  
    number = {1},
    publisher = {American Mathematical Society},
    pages = {1 -- 46},
    year = {1989},
}

@ARTICLE{Blum04CompoverReals,
    author = {Lenore Blum},
    title = {Computing over the Reals: Where {T}uring meets {N}ewton},
    journal = {Not. Am. Math. Soc.},
    year = {2004},
    volume = {51},
    number = {9},
    pages = {1024--1034}
}

@book{Blum98ComplRealComp,
    author = {Lenore Blum and Felipe Cucker and Mike Shub and Steve Smale},
    title = {Complexity and Real Computation},
    year = {1998},
    publisher = {Springer Verlag},
    address = {New York}
}

@book{Bochnak98RealAlgGeom,
    author = {Jacek Bochnak and Michel Coste and Marie-Francoise Roy},
    title = {Real Algebraic Geometry},
    year = {1998},
    publisher = {Springer Verlag},
    address = {Berlin, Heidelberg}
}

@book{Basu98AlgRealAlgGeom,
    author = {Saugata Basu and Richard Pollack and Marie-Francoise Roy},
    title = {Algorithms in Real Algebraic Geometry},
    year = {2006},
    publisher = {Springer Verlag},
    address = {Berlin, Heidelberg},
    edition = {2}
}

@book{Tarski51ProjThm,
    author = {Tarski, Alfred},
    title = {A {D}ecision {M}ethod for {E}lementary {A}lgebra and {G}eometry},
    year = {1951},
    publisher = {RAND Corporation}
}

@article{Seidenberg54ProjThm,
  title = {A NEW DECISION METHOD FOR ELEMENTARY ALGEBRA},
  author = {Abraham Seidenberg},
  journal = {Ann. Math.},
  year = {1954},
  volume = {60},
  pages = {365--374},
  number = {2},
}

@inproceedings{Cucker2019CompSemiAlg,
    booktitle = {Computing with Foresight and Industry},
    year = {2019},
    publisher = {Springer-Verlag},
    address = {Berlin, Heidelberg},
    author = {Felipe Cucker},
    title = {Recent Advances in the Computation of the Homology of Semialgebraic Sets},
    editor = {Florin Manea and Barnaby Martin and Daniel Paulusma and Giuseppe Primiero},
    pages = {1--12}
}

@article{Candes06Stable,
    author = {Candes, Emmanuel J. and Romberg, Justin K. and Tao, Terence},
    title = {Stable Signal Recovery from Incomplete and Inaccurate Measurements},
    journal = {Commun. Pure Appl. Math.},
    volume = {59},
    number = {8},
    pages = {1207-1223},
    year = {2006}
}

@article{Hilbert00Problems,
    author = {David Hilbert},
    title = {{Mathematical problems}},
    volume = {8},
    journal = {Bull. Am. Math. Soc.},
    number = {10},
    publisher = {American Mathematical Society},
    pages = {437 -- 479},
    year = {1902},
}

@article{Matiyasevich70Diophantine,
    author = {Matiyasevich, Yuri V. },
    title = {Enumerable sets are Diophantine},
    volume = {11},
    journal = {Soviet Mathematics},
    number = {2},
    pages = {354 -- 357},
    year = {1970},
}

@article{Daubechies04SparseReg,
    author = {Daubechies, I. and Defrise, M. and De Mol, C.},
    title = {An iterative thresholding algorithm for linear inverse problems with a sparsity constraint},
    journal = {Commun. Pure Appl. Math.},
    volume = {57},
    number = {11},
    pages = {1413-1457},
    year = {2004}
}

@ARTICLE{Candes06UnivEncStrat,
  author={Candes, Emmanuel J. and Tao, Terence},
  journal={IEEE Trans. Inf. Theory }, 
  title={Near-Optimal Signal Recovery From Random Projections: Universal Encoding Strategies?}, 
  year={2006},
  volume={52},
  number={12},
  pages={5406-5425},
}

@ARTICLE{Candes05DecLP,  
    author={Candes, E.J. and Tao, T.},
    journal={IEEE Trans. Inf. Theory },
    title={Decoding by linear programming},  
    year={2005},
    volume={51},
    number={12},
    pages={4203-4215},
}

@article{Candes06RobUnc,
  author={Candes, E.J. and Romberg, J. and Tao, T.},
  journal={IEEE Trans. Inf. Theory }, 
  title={Robust uncertainty principles: exact signal reconstruction from highly incomplete frequency information}, 
  year={2006},
  volume={52},
  number={2},
  pages={489-509},
}

@article{Donoho06CompSens,
  author={Donoho, D.L.},
  journal={IEEE Trans. Inf. Theory }, 
  title={Compressed sensing}, 
  year={2006},
  volume={52},
  number={4},
  pages={1289-1306},
}

@article{Zu18AutoMap,
    title	= {Image reconstruction by domain-transform manifold learning},
    author	= {Bo Zhu and Jeremiah Z. Liu and Stephen F. Cauley and Bruce R. Rosen and Matthew S. Rosen},
    year	= {2018},
    journal	= {Nature},
    pages	= {487--492},
    volume	= {555}
}

@ARTICLE{Schlemper18CNNvsCS,
  author={Schlemper, Jo and Caballero, Jose and Hajnal, Joseph V. and Price, Anthony N. and Rueckert, Daniel},
  journal={IEEE Trans. Med. Imaging}, 
  title={A Deep Cascade of Convolutional Neural Networks for Dynamic {MR} Image Reconstruction}, 
  year={2018},
  volume={37},
  number={2},
  pages={491-503},
}

@article{Arridge2019SolvingIP,
  title={Solving inverse problems using data-driven models},
  author={Simon Robert Arridge and Peter Maass and Ozan {\"O}ktem and Carola-Bibiane Sch{\"o}nlieb},
  journal={Acta Numerica},
  year={2019},
  volume={28},
  pages={1 - 174}
}

@article{Bubba19Shearlet,
    title = {Learning the Invisible: A Hybrid Deep Learning-Shearlet Framework for Limited Angle Computed Tomography},
    author = {Bubba, Tatiana A. and Gitta Kutyniok and Matti Lassas and Maximilian M{\"a}rz and Wojciech Samek and Samuli Siltanen and Vignesh Srinivasan},
    year = {2019},
    volume = {35},
    journal = {Inverse Problems},
    publisher = {IOP Publishing},
    number = {6},
}

@article{Hammernik18MRIDL2,
    author = {Hammernik, Kerstin and Klatzer, Teresa and Kobler, Erich and Recht, Michael P. and Sodickson, Daniel K. and Pock, Thomas and Knoll, Florian},
    title = {Learning a variational network for reconstruction of accelerated {MRI} data},
    journal = {Magn. Reson. Med.},
    volume = {79},
    number = {6},
    pages = {3055-3071},
    year = {2018}
}

@ARTICLE{Borgerding2017AMP,
  author={Borgerding, Mark and Schniter, Philip and Rangan, Sundeep},
  journal={IEEE Trans. Signal Process.}, 
  title={{AMP}-Inspired Deep Networks for Sparse Linear Inverse Problems}, 
  year={2017},
  volume={65},
  number={16},
  pages={4293-4308},
}

@INPROCEEDINGS{Boche2021DoSAttacks,
  author={Boche, Holger and Schaefer, Rafael F. and Vincent Poor, H.},
  booktitle={ICASSP 2021}, 
  title={Real Number Signal Processing can Detect Denial-of-Service Attacks}, 
  year={2021},
  pages={4765-4769},
  publisher={IEEE}
}

@INPROCEEDINGS{Boche2022RemoteState,
  author={Boche, Holger and B{ö}ck, Yannik and  Deppe, Christian},
  booktitle={ICC 2022}, 
  title={Deciding the Problem of Remote State Estimation via Noisy Communication Channels on Real Number Signal Processing Hardware}, 
  pages={4510–4515},
  publisher={IEEE}
}

@INPROCEEDINGS{Boche2021DetectDoS,
  author={Boche, Holger and Cai, Minglai and Poor, H. Vincent and Schaefer, Rafael F.},
  booktitle={ISIT 2021}, 
  title={Detectability of Denial-of-Service Attacks on Arbitrarily Varying Classical-Quantum Channels}, 
  year={2021},
  pages={912-917},
  publisher={IEEE}
}

@article{colbrook21stable,
    author = {Matthew J. Colbrook  and Vegard Antun  and Anders C. Hansen },
    title = {The difficulty of computing stable and accurate neural networks: {O}n the barriers of deep learning and {S}male’s 18th problem},
    journal = {Proc. Natl. Acad. Sci.},
    volume = {119},
    number = {12},
    year = {2022},
}

@misc{WebsiteIntel,
    author= {{Intel}},
    title = {Neuromorphic {C}omputing - {N}ext {G}eneration of {AI}},
    note = {\url{https://www.intel.com/content/www/us/en/research/neuromorphic-computing.html}, accessed on 2024-01-12},
}

@misc{WebsiteIBM,
    author= {{IBM Research Zurich}},
    title = {Neuromorphic {D}evices {\&} {S}ystems},
    note = {\url{https://www.zurich.ibm.com/st/neuromorphic/},  accessed on 2024-01-12},
}

@article{Ham2021SamsungNMC,
    author = {Ham, Donhee and Park, Hongkun and Hwang, Sungwoo and Kim, Kinam},
    year = {2021},
    pages = {635-644},
    title = {Neuromorphic electronics based on copying and pasting the brain},
    volume = {4},
    journal = {Nature Electronics},
}

@article{Araya18DLtomography,
    author = { Mauricio Araya-Polo  and  Joseph Jennings  and  Amir Adler  and  Taylor Dahlke },
    title = {Deep-learning tomography},
    journal = {The Leading Edge},
    volume = {37},
    number = {1},
    pages = {58-66},
    year = {2018},
}

@inproceedings{Chen2018LowLightPhoto,
  title={Learning to See in the Dark},
  author={Chen, Chen and Chen, Qifeng and Xu, Jia and Koltun, Vladlen},
  booktitle={CVPR 2018},
  publisher={IEEE}
}

@article{Rivenson17DLmicroscopy,
    author = {Yair Rivenson and Zolt\'{a}n G\"{o}r\"{o}cs and Harun G\"{u}naydin and Yibo Zhang and Hongda Wang and Aydogan Ozcan},
    journal = {Optica},
    number = {11},
    pages = {1437--1443},
    publisher = {OSA},
    title = {Deep learning microscopy},
    volume = {4},
    year = {2017},
}

@inproceedings{Yang16MRIDL,
    author = {Yang, Yan and Sun, Jian and Li, Huibin and Xu, Zongben},
    booktitle = {NIPS 2016},
    editor = {D. Lee and M. Sugiyama and U. Luxburg and I. Guyon and R. Garnett},
    publisher = {Curran Associates, Inc.},
    title = {Deep {ADMM}-Net for Compressive Sensing {MRI}},
    volume = {29},
}

@ARTICLE{Jin17DCNNInvProb,
  author={Jin, Kyong Hwan and McCann, Michael T. and Froustey, Emmanuel and Unser, Michael},
  journal={IEEE Trans. Image Process.}, 
  title={Deep Convolutional Neural Network for Inverse Problems in Imaging}, 
  year={2017},
  volume={26},
  number={9},
  pages={4509-4522},
}

@article{Adler17DNNInvProb,
	year = 2017,
	publisher = {{IOP} Publishing},
	volume = {33},
	number = {12},
	pages = {124007},
	author = {Jonas Adler and Ozan Öktem},
	title = {Solving ill-posed inverse problems using iterative deep neural networks},
	journal = {Inverse Problems},
}

@article{Ongie20DLInvProb,
  author={Ongie, Gregory and Jalal, Ajil and Metzler, Christopher A. and Baraniuk, Richard G. and Dimakis, Alexandros G. and Willett, Rebecca},
  journal={IEEE J. Sel. Areas Inf. Theory}, 
  title={Deep Learning Techniques for Inverse Problems in Imaging}, 
  year={2020},
  volume={1},
  number={1},
  pages={39-56},
}

@INPROCEEDINGS{Mousavi15DLvCSSigRec,
  author={Mousavi, Ali and Patel, Ankit B. and Baraniuk, Richard G.},
  booktitle={Allerton Conference 2015}, 
  title={A deep learning approach to structured signal recovery}, 
  pages={1336-1343},
}

@book{Goodfellow16DL,
    title={Deep Learning},
    author={Ian Goodfellow and Yoshua Bengio and Aaron Courville},
    publisher={MIT Press},
    note={\url{http://www.deeplearningbook.org}},
    year={2016}
}

@article{Turing36Entscheidung,
    author = {Turing, A. M.},
    title = {On Computable Numbers, with an Application to the {E}ntscheidungs-problem},
    journal = {Proc. Lond. Math. Soc.},
    volume = {s2-42},   
    number = {1},
    pages = {230-265},
    year = {1936}
}

@article{Borel1912Comp,
    author = {Emile Borel},
    journal = {Journal de Mathématiques Pures et Appliquées},
    pages = {159-210},
    title = {Le calcul des intégrales définies},
    volume = {8},
    year = {1912},
}

@article{Christensen2022NCSurvey,
	author={Christensen, Dennis Valbjørn and Dittmann, Regina and Linares-Barranco, Bernabe and Sebastian, Abu and Le Gallo, Manuel and Redaelli, Andrea and Slesazeck, Stefan and Mikolajick, Thomas and Spiga, Sabina and Menzel, Stephan and Valov, Ilia and Milano, Gianluca and Ricciardi, Carlo and Liang, Shi-Jun and Miao, Feng and Lanza, Mario and Quill, Tyler J. and Keene, Scott Tom and Salleo, Alberto and Grollier, Julie and Markovic, Danijela and Mizrahi, Alice and Yao, Peng and Yang, J. Joshua and Indiveri, Giacomo and Strachan, John Paul and Datta, Suman and Vianello, Elisa and Valentian, Alexandre and Feldmann, Johannes and Li, Xuan and Pernice, Wolfram HP and Bhaskaran, Harish and Furber, Steve and Neftci, Emre and Scherr, Franz and Maass, Wolfgang and Ramaswamy, Srikanth and Tapson, Jonathan and Panda, Priyadarshini and Kim, Youngeun and Tanaka, Gouhei and Thorpe, Simon and Bartolozzi, Chiara and Cleland, Thomas A and Posch, Christoph and Liu, Shih-Chii and Panuccio, Gabriella and Mahmud, Mufti and Mazumder, Arnab Neelim and Hosseini, Morteza and Mohsenin, Tinoosh and Donati, Elisa and Tolu, Silvia and Galeazzi, Roberto and Christensen, Martin Ejsing and Holm, Sune and Ielmini, Daniele and Pryds, Nini},
	title={2022 {R}oadmap on Neuromorphic Computing and Engineering},
	journal={Neuromorph. Comput. Eng.},
	year={2022},
	volume = {2},
	number = {2},
}

@article{Wagenbauer2017DNAassembly,
    title	= {Gigadalton-scale shape-programmable {DNA} assemblies},
    author	= {Wagenbauer, Klaus F. and Sigl, Christian and Dietz, Hendrik},
    year	= {2017},
    journal	= {Nature},
    pages	= {78--83},
    volume	= {552}
}

@article{Poirazi2017Dendritic,
    title	= {Illuminating dendritic function with computational models},
    author	= {Poirazi, Panayiota and Papoutsi, Athanasia},
    year	= {2020},
    journal	= {Nat. Rev. Neurosci.},
    pages	= {303--321},
    volume	= {21}
}

@article{Papp2021NanoscaleNN,
    title	= {Nanoscale neural network using non-linear spin-wave interference},
    author	= {Papp, \'{A}d\'{a}m and Porod, Wolfgang and Csaba, Gyorgy},
    year	= {2021},
    journal	= {Nat. Commun.},
    paper	= {6422}, 
    volume	= {12}
}

@article{Wright2022DeepPhysycalNN,
    title	= {Deep physical neural networks trained with backpropagation},
    author	= {Wright, Logan G. and Onodera, Tatsuhiro and Stein, Martin M. and Wang, Tianyu and Schachter, Darren T. and Hu, Zoey and McMahon, Peter L.},
    year	= {2022},
    journal	= {Nature},
    pages	= {549--555},
    volume	= {601}
}

@book{Ko91ComplTheoryRealFunc,
    author = {Ko, Ker-I},
    title = {Complexity Theory of Real Functions},
    year = {1991},
    publisher = {Birkhauser Boston Inc.},
    address = {USA}
}

@article{Belloni11SquareRootLasso,
    author = {Belloni, A. and Chernozhukov, V. and Wang, L.},
    title = {Square-Root Lasso: Pivotal Recovery of Sparse Signals via Conic Programming},
    journal = {Biometrika},
    volume = {98},
    number = {4},
    pages = {791-806},
    year = {2011},
}

@article{Tropp06Relax,
  author={Tropp, J.A.},
  journal={IEEE Trans. Inf. Theory }, 
  title={Just relax: convex programming methods for identifying sparse signals in noise}, 
  year={2006},
  volume={52},
  number={3},
  pages={1030-1051},
}

@article{Soare87RecursivelyES,
  title={Recursively enumerable sets and degrees},
  author={Robert Irving Soare},
  journal={Bull. Am. Math. Soc.},
  year={1987},
  volume={84},
  pages={1149-1181}
}

@book{Weihrauch00CompAnal,
    author = {Weihrauch, Klaus},
    title = {Computable Analysis: An Introduction},
    year = {2000},
    publisher = {Springer-Verlag},
    address = {Berlin, Heidelberg}
}

@book{Pour-El17Computability,
    AUTHOR = {Pour-El, Marian B. and Richards, J. Ian},
     TITLE = {Computability in {A}nalysis and {P}hysics},
    SERIES = {Perspectives in Logic},
 PUBLISHER = {Cambridge University Press},
      YEAR = {2017},
}

@inbook{AvigadBrattka14CompAnal,
    place={Cambridge}, 
    series={Lecture Notes in Logic}, 
    title={Computability and analysis: the legacy of {A}lan {T}uring}, 
    booktitle={Turing's Legacy: Developments from Turing's Ideas in Logic},
    publisher={Cambridge University Press}, 
    author={Avigad, Jeremy and Brattka, Vasco}, 
    editor={Downey, RodEditor}, 
    year={2014},
    pages={1–47},
    collection={Lecture Notes in Logic}
}

@article{Rumelhart86BP,
  title={Learning representations by back-propagating errors},
  author={David E. Rumelhart and Geoffrey E. Hinton and Ronald J. Williams},
  journal={Nature},
  year={1986},
  volume={323},
  pages={533-536}
}

@ARTICLE{Wright2009SparseRec,
  author={Wright, Stephen J. and Nowak, Robert D. and Figueiredo, MÁrio A. T.},
  journal={IEEE Trans. Signal Process.}, 
  title={Sparse Reconstruction by Separable Approximation}, 
  year={2009},
  volume={57},
  number={7},
  pages={2479-2493},
}

@ARTICLE{Ji2008BayesianCS,
  author={Ji, Shihao and Xue, Ya and Carin, Lawrence},
  journal={IEEE Trans. Signal Process.}, 
  title={Bayesian Compressive Sensing}, 
  year={2008},
  volume={56},
  number={6},
  pages={2346-2356},
}

@ARTICLE{Cotter2005SparseSol,
  author={Cotter, S.F. and Rao, B.D. and Kjersti Engan and Kreutz-Delgado, K.},
  journal={IEEE Trans. Signal Process.}, 
  title={Sparse solutions to linear inverse problems with multiple measurement vectors}, 
  year={2005},
  volume={53},
  number={7},
  pages={2477-2488},
}

@ARTICLE{Duarte2011StructuredCS,
  author={Duarte, Marco F. and Eldar, Yonina C.},
  journal={IEEE Trans. Signal Process.}, 
  title={Structured Compressed Sensing: from Theory to Applications}, 
  year={2011},
  volume={59},
  number={9},
  pages={4053-4085},
}

@ARTICLE{Selesnick2017SRviaCA,
  author={Selesnick, Ivan},
  journal={IEEE Trans. Signal Process.}, 
  title={Sparse Regularization via Convex Analysis}, 
  year={2017},
  volume={65},
  number={17},
  pages={4481-4494},
}

@ARTICLE{Lv2011GroupLasso,
  author={Lv, Xiaolei and Bi, Guoan and Wan, Chunru},
  journal={IEEE Trans. Signal Process.}, 
  title={The Group Lasso for Stable Recovery of Block-Sparse Signal Representations}, 
  year={2011},
  volume={59},
  number={4},
  pages={1371-1382},
}

@ARTICLE{Elad2007ProjCS,
  author={Elad, Michael},
  journal={IEEE Trans. Signal Process.}, 
  title={Optimized Projections for Compressed Sensing}, 
  year={2007},
  volume={55},
  number={12},
  pages={5695-5702},
}

@article{Chen1998AtomicDec,
    author = {Chen, Scott Shaobing and Donoho, David L. and Saunders, Michael A.},
    title = {Atomic Decomposition by Basis Pursuit},
    journal = {SIAM J. Sci. Comput.},
    volume = {20},
    number = {1},
    pages = {33-61},
    year = {1998},
}

@inproceedings{He2015DelvingDI,
  title={Delving Deep into Rectifiers: Surpassing Human-Level Performance on {I}mage{N}et Classification},
  author={Kaiming He and X. Zhang and Shaoqing Ren and Jian Sun},
  booktitle={ICCV 2015},
  publisher={IEEE},
  pages={1026-1034}
}

@article{Silver16Go,
    title	= {Mastering the game of {G}o with deep neural networks and tree search},
    author	= {David Silver and Aja Huang and Christopher J. Maddison and Arthur Guez and Laurent Sifre and George van den Driessche and Julian Schrittwieser and Ioannis Antonoglou and Veda Panneershelvam and Marc Lanctot and Sander Dieleman and Dominik Grewe and John Nham and Nal Kalchbrenner and Ilya Sutskever and Timothy Lillicrap and Madeleine Leach and Koray Kavukcuoglu and Thore Graepel and Demis Hassabis},
    year	= {2016},
    journal	= {Nature},
    pages	= {484--503},
    volume	= {529}
}

@article{Senior20DeepFold,
    title	= {Improved protein structure prediction using potentials from deep learning},
    author	= {Andrew W. Senior and Richard Evans and John Jumper and James Kirkpatrick and Laurent Sifre and Tim Green and Chongli Qin and Augustin Žídek and Alexander W. R. Nelson and Alex Bridgland and Hugo Penedones and Stig Petersen and Karen Simonyan and Steve Crossan and Pushmeet Kohli and David T. Jones and David Silver and Koray Kavukcuoglu and Demis Hassabis},
    year	= {2020},
    journal	= {Nature},
    pages	= {706--710},
    volume	= {577}
}

@inproceedings{Brown20GPT3,
    author = {Brown, Tom and Mann, Benjamin and Ryder, Nick and Subbiah, Melanie and Kaplan, Jared D and Dhariwal, Prafulla and Neelakantan, Arvind and Shyam, Pranav and Sastry, Girish and Askell, Amanda and Agarwal, Sandhini and Herbert-Voss, Ariel and Krueger, Gretchen and Henighan, Tom and Child, Rewon and Ramesh, Aditya and Ziegler, Daniel and Wu, Jeffrey and Winter, Clemens and Hesse, Chris and Chen, Mark and Sigler, Eric and Litwin, Mateusz and Gray, Scott and Chess, Benjamin and Clark, Jack and Berner, Christopher and McCandlish, Sam and Radford, Alec and Sutskever, Ilya and Amodei, Dario},
    booktitle = {NeurIPS 2020},
    editor = {H. Larochelle and M. Ranzato and R. Hadsell and M. F. Balcan and H. Lin},
    pages = {1877--1901},
    publisher = {Curran Associates, Inc.},
    title = {Language Models are Few-Shot Learners},
    volume = {33},
}

@article{Antun2020InstabilitiesDL,
	author = {Antun, Vegard and Renna, Francesco and Poon, Clarice and Adcock, Ben and Hansen, Anders C.},
	title = {On instabilities of deep learning in image reconstruction and the potential costs of {AI}},
	volume = {117},
	number = {48},
	pages = {30088--30095},
	year = {2020},
	publisher = {National Academy of Sciences},
	journal = {Proc. Natl. Acad. Sci.}
}

@inproceedings{Szegedy14AdvEx,
    title = {Intriguing properties of neural networks},
    author = {Christian Szegedy and Wojciech Zaremba and Ilya Sutskever and Joan Bruna and Dumitru Erhan and Ian Goodfellow and Rob Fergus},
    editor = {Yoshua Bengio and Yann LeCun},
    journal = {International Conference on Learning Representations},
    booktitle = {ICLR 2014}
}

@INPROCEEDINGS{Carlini18AudioAdvEx,
  author={Carlini, Nicholas and Wagner, David},
  booktitle={SPW 2018}, 
  title={Audio Adversarial Examples: Targeted Attacks on Speech-to-Text}, 
  volume={},
  number={},
  pages={1-7},
  publisher={IEEE}
}

@inproceedings{Tsipras18RobustnessOdds,
    title={Robustness May Be at Odds with Accuracy},
    author={Dimitris Tsipras and Shibani Santurkar and Logan Engstrom and Alexander Turner and Aleksander Madry},
    booktitle={ICLR 2019},
}

@inproceedings{Madry18AdvTraining,
  author    = {Aleksander Madry and Aleksandar Makelov and Ludwig Schmidt and Dimitris Tsipras and Adrian Vladu},
  title     = {Towards Deep Learning Models Resistant to Adversarial Attacks},
  booktitle = {ICLR 2018},
}

@INPROCEEDINGS {Papernot16Distillation,
    author = {N. Papernot and P. McDaniel and X. Wu and S. Jha and A. Swami},
    booktitle = {2016 IEEE Symposium on Security and Privacy},
    title = {Distillation as a Defense to Adversarial Perturbations Against Deep Neural Networks},
    volume = {},
    pages = {582-597},
    publisher = {IEEE Computer Society},
}

@inbook{Ilyas19AdvExNotBugs,
    author = {Ilyas, Andrew and Santurkar, Shibani and Tsipras, Dimitris and Engstrom, Logan and Tran, Brandon and Madry, Aleksander},
    title = {Adversarial Examples Are Not Bugs, They Are Features},
    publisher = {Curran Associates Inc.},
    address = {Red Hook, NY, USA},
    booktitle = {NeurIPS 2019},
}

@article{Mireshghallah20PrivacySurvey,
  author = {Mireshghallah, Fatemehsadat and Taram, Mohammadkazem and Vepakomma, Praneeth and Singh, Abhishek and Raskar, Ramesh and Esmaeilzadeh, Hadi},
  title = {Privacy in Deep Learning: A Survey},
  journal={arXiv:2004.12254},
  year = {2020},
}

@article{BOULEMTAFES20Privacy,
    title = {A review of privacy-preserving techniques for deep learning},
    journal = {Neurocomputing},
    volume = {384},
    pages = {21-45},
    year = {2020},
    author = {Amine Boulemtafes and Abdelouahid Derhab and Yacine Challal},
}

@ARTICLE{Liu21Privacy,
  author={Liu, Ximeng and Xie, Lehui and Wang, Yaopeng and Zou, Jian and Xiong, Jinbo and Ying, Zuobin and Vasilakos, Athanasios V.},
  journal={IEEE Access}, 
  title={Privacy and Security Issues in Deep Learning: A Survey}, 
  year={2021},
  volume={9},
  pages={4566-4593},
}

@InProceedings{Willers20Safety,
    author={Willers, Oliver and Sudholt, Sebastian and Raafatnia, Shervin and Abrecht, Stephanie},
    editor={Casimiro, Ant{\'o}nio and Ortmeier, Frank and Schoitsch, Erwin and Bitsch, Friedemann and Ferreira, Pedro},
    title={Safety Concerns and Mitigation Approaches Regarding the Use of Deep Learning in Safety-Critical Perception Tasks},
    booktitle={SAFECOMP 2020 Workshops},
    publisher={Springer International Publishing},
    address={Cham},
    pages={336--350},
}

@ARTICLE{He2022Security,
  author={He, Yingzhe and Meng, Guozhu and Chen, Kai and Hu, Xingbo and He, Jinwen},
  journal={IEEE Trans. Softw. Eng. }, 
  title={Towards Security Threats of Deep Learning Systems: A Survey}, 
  year={2022},
  volume={48},
  number={5},
  pages={1743-1770},
}

@article{Liu2020ComputingSF,
  title={Computing Systems for Autonomous Driving: State of the Art and Challenges}, 
  author={Liangkai Liu and Sidi Lu and Ren Zhong and Baofu Wu and Yongtao Yao and Qingyan Zhang and Weisong Shi},
  journal={IEEE Internet Things J.},
  year={2021},
  volume={8},
  number={8},
  pages={6469-6486}
}

@ARTICLE{Muhammad2021AutonomVehicle2,
    author={Muhammad, Khan and Ullah, Amin and Lloret, Jaime and Ser, Javier Del and de Albuquerque, Victor Hugo C.},
    journal={IEEE Trans. Intell. Transp. Syst.}, 
    title={Deep Learning for Safe Autonomous Driving: Current Challenges and Future Directions}, 
    year={2021}, 
    volume={22}, 
    number={7}, 
    pages={4316-4336},
}

@article{WU22Hitl,
    title = {A survey of human-in-the-loop for machine learning},
    journal = {Future Gener. Comput. Syst.},
    volume = {135},
    pages = {364-381},
    year = {2022},
    author = {Xingjiao Wu and Luwei Xiao and Yixuan Sun and Junhang Zhang and Tianlong Ma and Liang He},
}

@inproceedings{Mirman21Certification,
    author = {Mirman, Matthew and H\"{a}gele, Alexander and Bielik, Pavol and Gehr, Timon and Vechev, Martin},
    title = {Robustness Certification with Generative Models},
    publisher = {Association for Computing Machinery},
    address = {New York, NY, USA},
    booktitle = {SIGPLAN PLDI 2021},
    pages = {1141–1154},
}

@ARTICLE{Biondi20Certification,
  author={Biondi, Alessandro and Nesti, Federico and Cicero, Giorgiomaria and Casini, Daniel and Buttazzo, Giorgio},
  journal={IEEE Embed. Syst. Lett.}, 
  title={A Safe, Secure, and Predictable Software Architecture for Deep Learning in Safety-Critical Systems}, 
  year={2020},
  volume={12},
  number={3},
  pages={78-82},
}

@InProceedings{Katz17Certification,
    author={Katz, Guy and Barrett, Clark and Dill, David L. and Julian, Kyle and Kochenderfer, Mykel J.},
    editor={Majumdar, Rupak and Kun{\v{c}}ak, Viktor},
    title={Reluplex: An Efficient SMT Solver for Verifying Deep Neural Networks},
    booktitle={Computer Aided Verification},
    year={2017},
    publisher={Springer International Publishing},
    address={Cham},
    pages={97--117},
}

@inproceedings{Zhang20Certification,
  author    = {Huan Zhang and Hongge Chen and Chaowei Xiao and Sven Gowal and Robert Stanforth and Bo Li and Duane S. Boning and Cho{-}Jui Hsieh},
  title     = {Towards Stable and Efficient Training of Verifiably Robust Neural Networks},
  booktitle = {ICLR 2020},
}

@inproceedings{Salman19Certification,
     author = {Salman, Hadi and Li, Jerry and Razenshteyn, Ilya and Zhang, Pengchuan and Zhang, Huan and Bubeck, Sebastien and Yang, Greg},
     booktitle = {NeurIPS 2019},
     editor = {H. Wallach and H. Larochelle and A. Beygelzimer and F. d\textquotesingle Alch\'{e}-Buc and E. Fox and R. Garnett},
     pages = {},
     publisher = {Curran Associates, Inc.},
     title = {Provably Robust Deep Learning via Adversarially Trained Smoothed Classifiers},
     volume = {32},
}

@inproceedings{Esser15Neuromorphic,
 author = {Esser, Steve K and Appuswamy, Rathinakumar and Merolla, Paul and Arthur, John V. and Modha, Dharmendra S},
 booktitle = {NIPS 2015},
 editor = {C. Cortes and N. Lawrence and D. Lee and M. Sugiyama and R. Garnett},
 publisher = {Curran Associates, Inc.},
 title = {Backpropagation for Energy-Efficient Neuromorphic Computing},
 volume = {28},
}

@ARTICLE{Smith22Neuromorphic,
  author={Smith, J. Darby and Hill, Aaron J. and Reeder, Leah E. and Franke, Brian C. and Lehoucq, Richard B. and Parekh, Ojas and Severa, William and Aimone, James B.},
  journal={Nat. Electron.}, 
  title={Neuromorphic scaling advantages for energy-efficient random walk computations}, 
  year={2022},
  volume={5},
  number={2},
  pages={102-112},
}

@ARTICLE{Maas22Neuromorphic,
  author={Rao, Arjun and Plank, Philipp and Wild, Andreas and Maass, Wolfgang},
  journal={ Nat. Mach. Intell.}, 
  title={A Long Short-Term Memory for {AI} Applications in Spike-based Neuromorphic Hardware}, 
  year={2022},
  volume={4},
  number={5},
  pages={467-479},
}

@ARTICLE{Markovic20Neuromorphic,
  author={Markovi{ć}, Danijela and Mizrahi, Alice and Querlioz, Damien and Grollier, Julie},
  journal={Nat. Rev. Phys.}, 
  title={Physics for neuromorphic computing}, 
  year={2020},
  volume={2},
  number={9},
  pages={499-510},
}

@INPROCEEDINGS{Blouw20Neuromorphic,
  author={Blouw, Peter and Eliasmith, Chris},
  booktitle={ICASSP 2020}, 
  title={Event-Driven Signal Processing with Neuromorphic Computing Systems}, 
  publisher={IEEE},
  pages={8534-8538},
}

@INPROCEEDINGS{Boybat21InMemory,
  author={Boybat, I. and Kersting, B. and Sarwat, S. Ghazi and Timoneda, X. and Bruce, R. L. and BrightSky, M. and Gallo, M. Le and Sebastian, A.},
  booktitle={IEDM 2021}, 
  title={Temperature sensitivity of analog in-memory computing using phase-change memory}, 
  publisher={IEEE},
}

@ARTICLE{Karunaratne20InMemory,
  author={Karunaratne, Geethany and Le Gallo, Manuel and Cherubini, Giovanni and Benini, Luca and Rahimi, Abbas and Sebastian, Abu},
  journal={Nature Electronics}, 
  title={In-memory hyperdimensional computing}, 
  year={2020},
  volume={3},
  number={6},
  pages={327-337},
}

@ARTICLE{Sebastian20InMemory,
  author={Sebastian, Abu and Le Gallo, Manuel and Khaddam-Aljameh, Riduan and Eleftheriou, Evangelos},
  journal={Nat. Nanotechnol.}, 
  title={Memory devices and applications for in-memory computing}, 
  year={2020},
  volume={15},
  number={7},
  pages={529-544},
}

@Article{Payvand19InMemory,
    author ={Payvand, Melika and Nair, Manu V. and Müller, Lorenz K. and Indiveri, Giacomo},
    title  ={A neuromorphic systems approach to in-memory computing with non-ideal memristive devices: from mitigation to exploitation},
    journal  ={Faraday Discuss.},
    year  ={2019},
    volume  ={213},
    pages  ={487-510},
    publisher  ={The Royal Society of Chemistry},
}

@article{Boche2022PseudoInverse,
    title={Non-Computability of the Pseudoinverse on Digital Computers},
    author={Holger Boche and Adalbert Fono and Gitta Kutyniok},
    year={2022},
    journal={arXiv:2212.02940},
}

@book{lorentz13bernstein,
  title={Bernstein Polynomials},
  author={Lorentz, G.G.},
  series={AMS Chelsea Publishing},
  year={2013},
  publisher={American Mathematical Society}
}

@article{Stewart77PseudoInv,
    author = {Stewart, G. W.},
    title = {On the Perturbation of Pseudo-Inverses, Projections and Linear Least Squares Problems},
    journal = {SIAM Review},
    volume = {19},
    number = {4},
    pages = {634-662},
    year = {1977},
}

@article{Sheng2010CompMP,
    author = {Xingping Sheng  and  Guoliang Chen },
    title = {A note of computation for {M}-{P} inverse $A^\dagger$},
    journal = {Int. J. Comput. Math.},
    volume = {87},
    number = {10},
    pages = {2235-2241},
    year  = {2010},
    publisher = {Taylor & Francis},
}

@Article{Bloch22Quantum,
    author ={Daley, Andrew J. and Bloch, Immanuel and Kokail, Christian and Flannigan, Stuart and Pearson, Natalie and Troyer, Matthias and Zoller, Peter},
    title  ={Practical quantum advantage in quantum simulationn},
    journal  ={Nature},
    year  ={2022},
    volume  ={607},
    pages  ={667-676},
    publisher  ={The Royal Society of Chemistry},
}

@article{Bloch22Quantum2,
    year = {2022},
    publisher = {IOP Publishing},
    volume = {7},
    number = {4},
    author = {S Flannigan and N Pearson and G H Low and A Buyskikh and I Bloch and P Zoller and M Troyer and A J Daley},
    title = {Propagation of errors and quantitative quantum simulation with quantum advantage},
    journal = {Quantum Sci. Technol.},
}

@article{bastounis21extended,
      title={The extended {S}male's 9th problem -- {O}n computational barriers and paradoxes in estimation, regularisation, computer-assisted proofs and learning}, 
      author={Alexander Bastounis and Anders C Hansen and Verner Vlačić},
      year={2021},
      journal={arXiv:2110.15734}
}

\appendix

\section{Algorithmic Non-Solvability in the Turing model}\label{sec:app}
\noindent
In this section, we show that inverse problems expressed through the optimization problem 
\begin{equation}\label{eq:AppProb}
    \argmin_{x \in \C^N} \norm[\ast]{x} \text{ such that } \norm[\ell_2]{Ax -y} \leq \varepsilon, \tag{$\text{bp}^\ast$}
\end{equation}
are not algorithmically solvable in the Turing model, where the norm $\norm[\ast]{\cdot}$ is given by
\begin{equation*}
    \norm[\ast]{x} \coloneqq \sum_i^N \abs{\Re(x_i)} + \abs{\Im(x_i)}.
\end{equation*}
In particular, the following result will be proven.
\begin{Theorem}\label{thm:bpAstNonSolv}
    For $\varepsilon \in (0,1)$, the problem described by $\Xi_{\text{bp}^\ast,m,N,\varepsilon}$ is not algorithmically solvable on a Turing machine.
\end{Theorem}
First, we introduce the necessary notions describing Turing-computable functions. 

\subsection{Preliminaries from Computation Theory}
\noindent
For a detailed presentation of the topic, we refer to \cite{Soare87RecursivelyES, Weihrauch00CompAnal, Pour-El17Computability, AvigadBrattka14CompAnal}. Here, we only give a very short overview, starting with the computability of (real) numbers on Turing machines.
\begin{Definition}
    A sequence $(r_k)_{k \in \N}$ of \textit{rational numbers} is \textit{computable}, if there exist three recursive functions $a, b, s : \N \to \N$ such that $b(k) \neq 0$ for all $k \in \N$ and
    \begin{equation*} 
        r_k = (-1)^{s(k)} \frac{a(k)}{b(k)} \qquad \text{ for all } k \in \N. 
    \end{equation*}
\end{Definition}

\begin{Definition}
\begin{enumerate}[(1)]
    \item[]
    \item A number $x \in \R$ is computable if there exists a computable sequence of rational numbers $(r_k)_{k \in \N}$ such that
    \begin{equation*} 
        \abs{r_k - x} \leq 2^{-k} \qquad \text{ for all } k \in\N.
    \end{equation*}
    \item A \textit{vector} $v \in \R^n$ is \textit{computable} if each of its components is computable.
    \item A sequence $(x_n)_{n\in \N} \subset \R$ is  \textit{computable} if there exists a computable double-indexed sequence $\{r_{n,k}\}_{n,k \in \N} \subset \Q$ such that
    \begin{equation*} 
        \abs{r_{n,k} - x_n} \leq 2^{-k} \qquad \text{ for all } k, n \in\N.
    \end{equation*}
\end{enumerate}
\end{Definition}
Next, we will define the computability of a function via Banach-Mazur computability. This is in fact the weakest form of computability on a Turing machine in our setting, i.e., if a function is not Banach–Mazur computable, then it is not computable with respect to any other reasonable notion of computability on a Turing machine.
\begin{Definition}[Banach-Mazur computability]
    A function $f: I \to \R^n_c$, $I \subset \R^m_c$, is said to be \textit{Banach-Mazur computable}, if $f$ maps computable sequences $(t_n)_{n\in\N}\subset I$ onto computable sequences $(f(t_n))_{n\in\N}\subset \R^n_c$. 
\end{Definition}

\subsubsection{Proof Sketch}
\noindent
To prove algorithmic non-solvability in \Cref{thm:bpAstNonSolv}, we will apply certain non-computability conditions for optimization problems which were established in \cite{Boche2022LimitsDL}. 
\begin{Lemma}\label{lemma:generalNonApprox}
    Consider for the optimization problem $P$ with optimization parameter $\mu>0$ the multi-valued mapping $\Xi_{P,m,N,\mu}$ defined by
    \begin{align*}
        \Xi_{P,m,N,\mu}: \C^{m\times N} \times \C^m &\rightrightarrows \C^N \\  
        (A,y) &\mapsto P(A,y,\mu).\nonumber
    \end{align*}
    Choose an arbitrary single-valued restriction $\Psi^s \in \mathcal{M}_{\Xi_{P,m,N,\mu}}$ and $\Omega \subseteq \dom(\Psi^s) = \{\omega =(A,y) \in \C^{m\times N} \times \C^m\} \,\vert\, P(\omega,\mu) \neq \emptyset \}$. Further, suppose that there are two computable sequences $(\omega_{n}^1)_{n\in \N}, (\omega_{n}^2)_{n\in \N} \subset \Omega$ satisfying the following conditions:
    \begin{enumerate}[(a)]
        \item There are sets $S^1, S^2 \subset \C^N$ and $\kappa > 0, \kappa \in \Q$ such that $\inf_{x_1 \in S^1, x_2\in S^2} \norm[\ell_2]{x_1 - x_2} > \kappa$ and $\Psi^s(\omega_{n}^j) \in S^j$ for $j = 1, 2$.
        \item There exists $\omega^\ast \in \Omega$ such that $\norm[\ell_2]{\omega_{n}^j - \omega^\ast} \leq 2^{-n}$ for all $n \in \N$, $j = 1, 2$. 
    \end{enumerate}
    In addition, let $\Psi: \Omega_\Psi \to \C^N$, $\Omega \subset \Omega_\Psi \subset \C^{m\times N} \times \C^m$, be an arbitrary function with
    \begin{equation*}
        \sup_{\omega \in \Omega} \norm[\ell_2]{\Psi^s(\omega) - \Psi(\omega)} < \frac{\kappa}{8}.
    \end{equation*}
    Then $\Psi$ is not Banach-Mazur computable.
\end{Lemma}
Verifying the assumptions of \Cref{lemma:generalNonApprox} for $\varepsilon\in (0,1)$ is sufficient to obtain \Cref{thm:bpAstNonSolv}. We will only sketch the approach to demonstrate algorithmic non-solvability in the Turing model, however, a more thorough analysis also yields algorithmic non-approximability for a certain range of optimization parameter. 
%The theorem implies that the mapping is not Banach-Mazur computable for any choice. Hence, the algorithmic non-solvability in the Turing model of the modified optimization problem \eqref{eq:AppProb} immediately follows if the conditions $(a)$ and $(b)$ in \Cref{thm:TuringMain} are satisfied. 
The construction of the sequences satisfying conditions $(a)$ and $(b)$ in \Cref{lemma:generalNonApprox} relies on a characterization of the solution set. The same properties of the solution set as in the original basis pursuit problem can be derived by adapting the proof technique in \cite[Appendix]{colbrook21stable}) to the \eqref{eq:AppProb} formulation.
%Note that in the theorem we did not specify which of the possibly multiple values the mapping $\Psi_{\varepsilon}$ takes for each input. The theorem implies that the mapping is not Banach-Mazur computable for any choice. Hence, the algorithmic non-solvability in the Turing model of the modified optimization problem \eqref{eq:AppProb} immediately follows if the conditions $(a)$ and $(b)$ in \Cref{thm:TuringMain} are satisfied. The construction of the described sequences relies on a characterization of the solution set. The same properties of the solution set can be obtained as in the original problem by applying the same proof technique (cf. \cite[appendix]{colbrook21stable}).
\begin{Lemma}\label{lm:pt2}
    Let $N \geq 2$ and consider \eqref{eq:AppProb} for
    \begin{equation*}
        A = \begin{pmatrix} a_1 & a_2 & \dots & a_N \end{pmatrix} \in \C^{1 \times N}, \quad y=1, \quad \varepsilon \in [0,1),
    \end{equation*}
    where $a_j > 0$ for $j=1,\dots,N$. Then the set of solutions of $\text{bp}^\ast$ is given by
    \begin{align*}
        \sum_{j=1}^{N} \left(t_j(1-\varepsilon) a_j^{-1} \right) e_j, \quad  s.t. \quad  &t_j \in [0,1], \quad \sum_{j=1}^{N} t_j = 1,\\
        &\text{ and } t_j=0 \textit{ if } a_j < \max_k a_k,    
    \end{align*}
    where $\{e_j\}_{j=1}^N$ is the canonical basis of $\C^N$.
\end{Lemma}
Finally, using \Cref{lm:pt2} the sequences establishing Banach-Mazur non-computability of any map $\Psi^s \in \Xi_{\text{bp}^\ast,m,N,\varepsilon}$ for $\varepsilon \in (0,1)$ can be constructed exactly as for \eqref{eq:sparseprob} in \cite[Proof of Theorem 4.1]{Boche2022LimitsDL} so that \Cref{thm:bpAstNonSolv} follows.

% \begin{Remark}
%     For $\varepsilon \in (0,1)$, we immediately conclude that there exists no Turing-computable function $\Xi^s \in \mathcal{M}_{\Xi_{\text{bp}^\ast,m,N,\varepsilon}}$. Hence, the problem described by $\Xi_{\text{bp}^\ast,m,N,\varepsilon}$ is not algorithmically solvable on a Turing machine.
% \end{Remark}

\section{Proof of \Cref{thm:BPapprox}}\label{sec:appb}
\noindent
In this section, we provide the proof of \Cref{thm:BPapprox}. First, we introduce some preliminary results starting with the Weierstrass approximation theorem. 
\begin{Theorem}\label{thm:Weierstrass}
    Let $f: [a,b] \to \R$ be a continuous function and $[a,b] \subset \R$ an interval. For every $\gamma >0$, there exists a polynomial $p_\gamma$ such that 
    \begin{equation*}
        \sup_{x\in [a,b]} \abs{f(x) - p_\gamma(x)} < \gamma.
    \end{equation*}
\end{Theorem}
\begin{Remark}\label{rm:Weierstrass}
    A constructive proof, see e.g. \cite{lorentz13bernstein}, of the Weierstrass approximation theorem can be established via Bernstein polynomials
    \begin{equation*}
        B_n(x,f) \coloneqq \sum_{k=0}^n f\Big(\frac{k}{n}\Big) \binom{n}{k} x^k (1-x)^{n-k} \text{ of degree } n \in \N
    \end{equation*}
    associated to a real-valued function $f \in C([0,1])$. In particular, for $\gamma>0$ we have for any $x\in[0,1]$ that
    \begin{equation}\label{eq:Bnapprox}
        \abs{B_n(x,f) - f(x)} \leq \gamma \quad \text{ if } n \geq \frac{\norm[\infty]{f}}{\delta^2 \gamma},
    \end{equation}
    where $\delta>0$ satisfies for all $x,y\in [0,1]$:
    \begin{equation*}
        \abs{x-y}\leq \delta \quad \implies \quad \abs{f(x) - f(y)} \leq \frac{\gamma}{2}.
    \end{equation*}
    Note that for any $\gamma>0$ there exists a corresponding $\delta>0$ due to the uniform continuity of $f$ on $[0,1]$. By considering the function $\Phi : [0,1] \to [a,b]$ given by
    \begin{equation*}
        \phi(x) = (b-a)x + a,
    \end{equation*}
    we can extend \eqref{eq:Bnapprox} to real-valued functions $g \in C([a,b])$ on arbitrary intervals $[a,b]$, i.e., for $\gamma>0$ we have for any $x\in[a,b]$ that
    \begin{equation*}
        \abs{B_n(\phi^{-1}(x),g \circ \phi) - g(x)} \leq \gamma \quad \text{ if } n \geq \frac{\norm[\infty]{g \circ \phi}}{\delta^2 \gamma},
    \end{equation*}    
    where $\delta>0$ satisfies for all $x,y\in [a,b]$:
    \begin{equation*}
        \abs{x-y}\leq \delta \quad \implies \quad \abs{(g \circ \phi)(x) - (g \circ \phi)(y)} \leq \frac{\gamma}{2}.
    \end{equation*}
\end{Remark}
Our goal is to approximate the objective of basis pursuit optimization \eqref{eq:sparseprob}, i.e., the $\ell_1$ norm, by a polynomial on a suitable domain $I \subset \C^N$. In particular, we want $I$ to contain the solutions of \eqref{eq:sparseprob} so that the optimization on the approximate objective ideally leads to a solution of the original problem, albeit this can not be guaranteed. After identifying $I$, we can invoke \Cref{thm:Weierstrass} and \Cref{rm:Weierstrass} to construct the sought polynomial. Hence, the first step is to find a domain satisfying our conditions. Via the pseudoinverse we are able to construct $I$ explicitly. In particular, we will use the following fact \cite{Stewart77PseudoInv}: For $A\in \C^{m\times N}$ and $y\in \C^m$ the minimal $\ell_2$ norm solution of
\begin{equation}\label{eq:PIsol}
    \argmin_{x\in \C^N} \norm[\ell_2]{Ax-y}
\end{equation}
is given by $A^\dagger y$, where $A^\dagger \in \C^{N\times m}$ denotes the pseudoinverse of $A$.
\begin{Lemma}\label{lm:PIint}
    Let $A\in \C^{m\times N}$, $y\in \C^m$ and $\varepsilon >0$ and consider
    \begin{equation}\label{eq:bpappend}
        \argmin_{x \in \C^N} \norm[\ell_1]{x} \text{ such that } \norm[\ell_2]{Ax -y} \leq \varepsilon.
    \end{equation}
    If the solution set is not empty, then the solution set is contained in $I=\{x\in\C^N: \norm[\ell_2]{x} < \sqrt{N} \norm[\ell_2]{A^\dagger y}\}$.    
\end{Lemma}
\begin{proof}
    Set $\hat{x}= A^\dagger y$. If $\norm[\ell_2]{A\hat{x}-y}> \varepsilon$, then \eqref{eq:PIsol} implies that there does not exist a minimizer of \eqref{eq:bpappend} for the given set of parameters. Hence, we assume $\norm[\ell_2]{A\hat{x}-y} \leq \varepsilon$ so that $\hat{x}$ satisfies the constraint of \eqref{eq:bpappend}. Note that $\norm[\ell_1]{x} \geq \norm[\ell_2]{x}$ and $\norm[\ell_1]{x} \leq  \sqrt{n} \norm[\ell_2]{x}$ for all $x \in \C^n$. Thus, $\norm[\ell_2]{z}> \sqrt{N} \norm[\ell_2]{\hat{x}}$ entails that
    \begin{equation*}
        \norm[\ell_1]{z} \geq \norm[\ell_2]{z} >  \sqrt{N} \norm[\ell_2]{\hat{x}} \geq \norm[\ell_1]{\hat{x}}
    \end{equation*}
    for $z\in \C^N$, i.e., $z$ is not a minimizer of \eqref{eq:bpappend}. Therefore, the optimization problem
    \begin{equation*}
        \argmin_{x \in \C^N} \norm[\ell_1]{x} \text{ such that } \norm[\ell_2]{Ax -y} \leq \varepsilon \,\,\wedge\,\, \norm[\ell_2]{x} < \sqrt{N} \norm[\ell_2]{\hat{x}}
    \end{equation*}    
    is equivalent to \eqref{eq:bpappend} and its minimizers are contained in $I$.
\end{proof}
Now, we turn to the approximation of the $\ell_1$ norm via polynomials.
\begin{Lemma}\label{lm:polAppr}
    Let $\beta>0$. For every $\gamma >0$, there exists a polynomial $p_{\beta,\gamma}$ such that 
    \begin{equation*}
        \sup_{x\in \C^n : \norm[\ell_2]{x} < \sqrt{n} \beta} \abs{\norm[\ell_1]{x} - p_{\beta,\gamma}(\Re(x), \Im(x))} < \gamma.
    \end{equation*}  
    The polynomial $p_{\beta,\gamma}$ can be explicitly expressed as
    \begin{equation}\label{eq:polBG}
        p_{\beta,\gamma}(\Re(x), \Im(x)) =  \frac{1}{(n\beta)^{2k-1}} \sum_{j=0}^k \sqrt{\frac{j}{k}}  \binom{k}{j} \sum_{i=1}^n \big(\Re(x_i)^2 + \Im(x_i)^2\big)^j \big(n^2\beta^2 -\Re(x_i)^2 - \Im(x_i)^2\big)^{k-j},
    \end{equation}
    where $ k \geq \frac{4n^8 \beta^5}{\gamma^5}$ is required.
\end{Lemma}
\begin{proof}
    Let $x\in \C^n$ such that $\norm[\ell_2]{x} < \sqrt{n} \beta$ holds. Then,
    \begin{equation*}
        \norm[\ell_1]{x} \leq \sqrt{n} \norm[\ell_2]{x} <  n \beta
    \end{equation*}
    so that
    \begin{equation*}
        n \beta > \norm[\ell_1]{x} = \sum_{i=1}^n \abs{x_i} = \sum_{i=1}^n \sqrt{\Re(x_i)^2 + \Im(x_i)^2}
    \end{equation*}
    and in particular
    \begin{equation}\label{eq:sqrtDom}
        n \beta > \sqrt{\Re(x_i)^2 + \Im(x_i)^2} \quad \text{ for all } i=1,\dots,n.
    \end{equation}
    Next, we want to invoke \Cref{thm:Weierstrass} and \Cref{rm:Weierstrass} to approximate $\sqrt{\cdot}$ on the interval $[0,n^2 \beta^2]$. To that end, let $g(x) = \sqrt{x}$ on $[0,n^2 \beta^2]$ and $\phi(x)= n^2 \beta^2 x$ on $[0,1]$. Observe that $(g\circ \phi)(x) = n \beta \sqrt{x}$ so that $\norm[\infty]{g \circ \phi} = n \beta$ and 
    \begin{equation*}
        \abs{n \beta \sqrt{x} - n \beta \sqrt{y}}^2 \leq n^2 \beta^2 \abs{\sqrt{x} - \sqrt{y}} \abs{\sqrt{x} + \sqrt{y}} =  n^2 \beta^2 \abs{x-y} \quad \text{ for all } x,y \in [0,1],
    \end{equation*}
    i.e., we obtain
    \begin{equation*}
        \abs{x-y} \leq \frac{1}{n^2 \beta^2} \Big(\frac{\gamma}{2n}\Big)^2 \eqqcolon \delta \implies \abs{(g \circ \phi)(x) - (g \circ \phi)(y)} \leq \frac{\gamma}{2n}.
    \end{equation*}
    Hence, applying \Cref{rm:Weierstrass} yields that for all $x,y \in [0,n^2 \beta^2]$
    \begin{equation*}
        \abs{B_k(\phi^{-1}(x),g \circ \phi) - g(x)} \leq \frac{\gamma}{n} \quad \text{ if } k \geq \frac{n \beta}{\delta^2 \frac{\gamma}{n}} =  \frac{4n^8 \beta^5}{\gamma^5}.
    \end{equation*}
    Thus, due to \eqref{eq:sqrtDom}
    \begin{align*}
        \sup_{x\in \C^n : \norm[\ell_2]{x} < \sqrt{n} \beta} &\abs{\norm[\ell_1]{x} - \sum_{i=1}^n B_k(\phi^{-1}(\Re(x_i)^2 + \Im(x_i)^2),g \circ \phi)}\\
        &=  \abs{\sum_{i=1}^n \sqrt{\Re(x_i)^2 + \Im(x_i)^2} - B_k(\phi^{-1}(\Re(x_i)^2 + \Im(x_i)^2),g \circ \phi)} \\
        &\leq \sum_{i=1}^n  \abs{ g(\Re(x_i)^2 + \Im(x_i)^2) - B_k(\phi^{-1}(\Re(x_i)^2 + \Im(x_i)^2),g \circ \phi)}\\
        &\leq n \cdot \frac{\gamma}{n} = \gamma \quad \text{ if } k \geq \frac{4n^8 \beta^5}{\gamma^5}
    \end{align*}
    and the sought polynomial $p_{\beta,\gamma}$ indeed exists via
    \begin{align*}
        p_{\beta,\gamma}(\Re(x)&, \Im(x)) =   \sum_{i=1}^n B_k(\phi^{-1}(\Re(x_i)^2 + \Im(x_i)^2),g \circ \phi)  \\
        &= \sum_{i=1}^n \sum_{j=0}^k (g \circ \phi)\Big(\frac{j}{k}\Big)  \binom{k}{j} \Big(\phi^{-1}\big(\Re(x_i)^2 + \Im(x_i)^2\big)\Big)^j \Big(1- \phi^{-1}\big(\Re(x_i)^2 + \Im(x_i)^2\big)\Big)^{k-j} \\
        &= \sum_{i=1}^n \sum_{j=0}^k n \beta \sqrt{\frac{j}{k}}  \binom{k}{j} \Big(\frac{1}{n^2\beta^2}(\Re(x_i)^2 + \Im(x_i)^2)\Big)^j \Big(1-\frac{1}{n^2\beta^2}(\Re(x_i)^2 + \Im(x_i)^2)\Big)^{k-j} \\
        &= \frac{1}{(n\beta)^{2k-1}} \sum_{j=0}^k \sqrt{\frac{j}{k}}  \binom{k}{j} \sum_{i=1}^n \big(\Re(x_i)^2 + \Im(x_i)^2\big)^j \big(n^2\beta^2 -\Re(x_i)^2 - \Im(x_i)^2\big)^{k-j}.
    \end{align*}
\end{proof}
So far, we have established the existence of an optimization problem approximating basis pursuit in the following sense: For given $A\in \C^N$, $y\in \C^m$ and $\varepsilon>0$   
\begin{itemize}
    \item \Cref{lm:PIint} describes a domain 
    \begin{equation}\label{eq:IAy}
        I_{A,y}\coloneqq \{x \in \C^N: \norm[\ell_2]{x} < \sqrt{N} \norm[\ell_2]{A^\dagger y}   \} 
    \end{equation}
    that contains the solutions of basis pursuit for $(A,y,\varepsilon)$,
    \item and \Cref{lm:polAppr} describes a polynomial $p_{\beta,\gamma}$ depending on parameters $\beta,\gamma>0$ that approximates the $\ell_1$ norm on a set 
    \begin{equation}\label{eq:Ibeta}
        I_\beta \coloneqq \{x \in \C^N: \norm[\ell_2]{x} < \sqrt{N} \beta  \}   
    \end{equation} 
    up to an error of $\gamma$.
\end{itemize}
Hence, an appropriate choice of $\beta = \norm[\ell_2]{A^\dagger y}$ enables to approximate the objective of basis pursuit up to an error of $\gamma$ on $I_{A,y}$. Note that the convergence of $p_{\beta,\gamma}$ to the $\ell_1$ norm for $\gamma \to 0$ does not imply that the minimizer(s) of the respective optimization problems also converge towards the original minimizer(s) of basis pursuit. Unfortunately, obtaining a guarantee of this form does not appear to be feasible in the general case. Thus, we have to be content with the convergence of the objectives. 

The remaining question is whether the adjusted optimization problem with additional parameters $\beta$ and $\gamma$ can be implemented and solved on BSS machines. We immediately observe that the function $p_{\beta,\gamma}$ given in \eqref{eq:polBG} can not be computed on a BSS machine. Indeed, $p_{\beta,\gamma}$ depends on coefficients encompassing values $(\sfrac{j}{k})^{\sfrac{1}{2}}$, where $j=0,\dots,k$ and $k\in \N$ needs to be chosen based on $\beta$ and $\gamma$. Thus, the non-computability of the square-root function prohibits the computation of the coefficients on a BSS machine. Therefore, the corresponding optimization problem can not be implemented on a BSS machine for adjustable parameters $\beta$ and $\gamma$. On the other hand, we will show next that in the restricted case with fixed $\beta$ and $\gamma$ the polynomial $p_{\beta,\gamma}$ can be implemented on a BSS machine. 
\begin{Lemma}\label{lm:polBSS}
    Fix $\beta, \gamma > 0$ and $k \geq \frac{4n^8 \beta^5}{\gamma^5}$. Then $p_{\beta,\gamma}$ as defined in \eqref{eq:polBG} is BSS-computable. 
\end{Lemma}
\begin{proof}
    Based on $k$ the finite set $\{\sqrt{1}, \dots, \sqrt{k}\} \cup \{\beta, \gamma\}$ can be encoded as constants in the memory of a BSS machine. Therefore, the coefficients of the polynomial $p_{\beta,\gamma}$ can be computed on a BSS machine so that $p_{\beta,\gamma}$ is a BSS-computable function as polynomials can be expressed through the concatenation of arithmetic operations.
\end{proof}
Finally, we prove that the optimization problem with objective $p_{\beta,\gamma}$ on $I_\beta$ can be solved on a BSS machine. Thereby, we conclude the proof of \Cref{thm:BPapprox}.
\begin{proof}[Proof of \Cref{thm:BPapprox}]
We first show that there exists a BSS-computable function $\Xi^s \in \mathcal{M}_{\Xi_{\text{p}(\beta,\gamma),m,N}}$ for fixed $\beta, \gamma >0$. We can manually compute $k\in \N$ such that $k \geq \frac{4N^8 \beta^5}{\gamma^5}$. Then, \Cref{lm:polBSS} implies that the polynomial $p_{\beta,\gamma}$ given in \eqref{eq:polBG} can be implemented on a BSS machine and \Cref{lm:polAppr} entails that \eqref{eq:approxl1} is satisfied. Additionally, observe that
\begin{equation*}
    \norm[\ell_2]{x} < \sqrt{N} \beta \, \Leftrightarrow \, \norm[\ell_2]{x}^2 < N \beta^2 \, \Leftrightarrow  \,\sum_{i=1}^n \abs{x_i}^2 - N \beta^2 < 0  \, \Leftrightarrow \, \sum_{i=1}^n \Re(x_i)^2 + \Im(x_i)^2 - N \beta^2 < 0   
\end{equation*}
is a semialgebraic set defined by a polynomial in $(\Re(x),\Im(x))$. Therefore \eqref{eq:approxbp} describes the optimization of the polynomial $p_{\beta,\gamma}$ on a semialgebraic set. Thus, the existence of $\Xi^s$ follows analogously to the proofs of \Cref{thm:BSSReal} and \Cref{thm:BSSComplex} via the optimization algorithm \eqref{alg:GlOp}.

It is left to show that there exists a BSS-computable function $g: \C^{m\times N} \times \C^m \times \R_{>0} \times \R_{>0}\to \{0,1\}$ so that $g(A,y, \varepsilon,\beta) = 1$ implies that the solutions of basis pursuit \eqref{eq:sparseprob} are contained in $I_\beta = \{x \in \C^N: \norm[\ell_2]{x} < \sqrt{N} \beta\}$. We will explicitly construct $g$. First, note that the function $g_{\text{pi}}: \C^{m\times N} \to \C^{N\times m}$ mapping $A$ on its pseudoinverse $g_{\text{pi}}(A)=A^\dagger$ is a BSS-computable function. This follows immediately from the fact that the pseudoinverse can be computed via multiple applications of the Gaussian elimination algorithm and exchange of rows of matrices based on identifying zero-rows \cite{Sheng2010CompMP}. These operations can all be implemented on BSS machines since they depend only on basic arithmetic operations and comparisons to zero. Hence, the computability of $g_{\text{pi}}$ follows. 

Next, define the set $S \coloneqq \{(A,y,\varepsilon) \in \C^{m\times N} \times \C^m \times \R_{>0} : g_{\text{sol}}(A,y,\varepsilon) < 0 \}$, where $g_{\text{sol}}: \C^{m\times N} \times \C^m \times \R_{>0}$ is given by
\begin{equation*}
    g_{\text{sol}}(A,y,\varepsilon) = \norm[\ell_2]{Ag_{\text{pi}}(A) y - y}^2 - \varepsilon^2.
\end{equation*}
We can immediately conclude that $g_{\text{sol}}$ is BSS-computable since $g_{\text{pi}}$ is BSS-computable and $g_{\text{sol}}$ is a polynomial in $(A,y,\varepsilon)$. Hence, $S$ is decidable on a BSS machine. Recall that $\norm[\ell_2]{Ag_{\text{pi}}(A) y - y}^2 - \varepsilon^2 > 0$ implies that the associated basis pursuit optimization does not have any solution; compare proof of \Cref{lm:PIint}. Consequently, the existence of minimizers of basis pursuit can be decided on BSS machines.  

Finally, given that minimizer(s) for basis pursuit exist, we need to check whether they are contained in $I_\beta$. Due to \eqref{eq:IAy} and \eqref{eq:Ibeta} it suffices to decide the set $V \coloneqq \{(A,y,\beta) \in \C^{m\times N} \times \C^m \times \R_{>0}: g_{\text{dec}}(A,y,\beta) \leq 0 \}$, where $g_{\text{dec}}: \C^{m\times N} \times \C^m \times \R_{>0}$ is given by
\begin{equation*}
    g_{\text{dec}}(A,y,\beta) = \norm[\ell_2]{g_{\text{pi}}(A) y}^2 - \beta^2.
\end{equation*}
Again the BSS-computability of $g_{\text{dec}}$ follows directly from the BSS-computability of $g_{\text{pi}}$ so that $V$ is decidable on a BSS machine. Therefore, setting 
\begin{equation*}
    g(A,y,\varepsilon,\beta) = \chi_S(A,y,\varepsilon) \chi_V(A,y,\beta)
\end{equation*}
gives that $g$ is BSS-computable as the product of two BSS-computable functions. Moreover, \eqref{eq:IAy} shows that $g(A,y,\varepsilon,\beta) = 1$ if and only if the solutions of basis pursuit for $(A,y,\varepsilon)$ are contained in $I_\beta$. Finally, the BSS-computability of $g$ remains valid when considering real BSS machines representing $\C$ via the real and imaginary parts. 
 
\end{proof}
\begin{Remark}
    In the Turing model the outlined approach to approximate basis pursuit is not feasible. In particular, the step involving the pseudoinverse of the input matrix is not transferable since the mapping of a matrix on its pseudoinverse is not computable on a Turing machine \cite{Boche2022PseudoInverse}.  
\end{Remark}
\end{document}